\documentclass[runningheads,a4paper]{llncs}

\usepackage[USenglish]{babel}

\usepackage{ltl}
\usepackage{wrapfig}
\usepackage{lipsum,booktabs}
\usepackage{booktabs}
\usepackage{tikz}
\usetikzlibrary{calc}
\usepackage{color}
\usepackage{scalefnt}
\usepackage{makecell}
\usepackage{array}
\usepackage{cleveref}

\usepackage{amssymb}
\usepackage{graphicx}
\usepackage[noend]{algorithmic}
\usepackage{algorithm2e}
\usepackage{stmaryrd}

\usepackage{xspace}
\usepackage{multirow}

\usepackage{tikz}
\usetikzlibrary{positioning}

\usepackage{url}
\urldef{\mailsa}\path|{alfred.hofmann, ursula.barth, ingrid.haas, frank.holzwarth,|
\urldef{\mailsb}\path|anna.kramer, leonie.kunz, christine.reiss, nicole.sator,|
\urldef{\mailsc}\path|erika.siebert-cole, peter.strasser, lncs}@springer.com|

\newcommand{\nats}{\mathbb{N}}

\newcommand{\qform}{\Phi}
\newcommand{\hform}{\chi}
\newcommand{\pathform}{\varphi}
\newcommand{\pexpr}{P}
\newcommand{\AP}{\mathsf{AP}}
\newcommand{\prob}{\mathbb{P}}
\newcommand{\Prob}{\mathit{Prob}}
\newcommand{\phl}{\text{PHL}\xspace}
\newcommand{\hyperltl}{\mbox{\sc HyperLTL}\xspace}
\newcommand{\hyperctl}{\mbox{\sc HyperCTL$^*$}\xspace}
\newcommand{\hyperpctl}{\mbox{\sc HyperPCTL}\xspace}

\newcommand{\Act}{\mathit{Act}}
\newcommand{\trans}{\mathbf{P}}
\newcommand{\init}{\iota}
\newcommand{\Paths}{\mathit{Paths}}
\newcommand{\dist}{\mathcal{D}}
\newcommand{\Sched}{\mathit{Sched}}
\newcommand{\DSched}{\mathit{DetSched}}
\newcommand{\scheduler}{\mathfrak{S}}

\newcommand{\Pathsf}{\mathit{Paths}_\mathit{fin}}
\newcommand{\Pathsi}{\mathit{Paths}_\mathit{inf}}
\newcommand{\cyl}{\mathit{Cyl}}

\newcommand{\sassign}{\Sigma}
\newcommand{\passign}{\Pi}
\newcommand{\varsp}{\mathcal{V}_{path}}
\newcommand{\varss}{\mathcal{V}_{sched}}

\begin{document}

\mainmatter  % start of an individual contribution

% first the title is needed
\title{Probabilistic Hyperproperties \\of Markov Decision Processes\thanks{This work was partially supported by the Collaborative Research Center “Foundations of Perspicuous Software Systems” (TRR 248, 389792660),  the European Research Council (ERC) Grant OSARES (No. 683300),  the DARPA Assured Autonomy program, the iCyPhy center, and by Berkeley Deep Drive.}}

% a short form should be given in case it is too long for the running head
\titlerunning{}

% the name(s) of the author(s) follow(s) next
%
% NB: Chinese authors should write their first names(s) in front of
% their surnames. This ensures that the names appear correctly in
% the running heads and the author index.
%
\author{Rayna Dimitrova \inst{1}\and Bernd Finkbeiner \inst{1} \and Hazem Torfah \inst{2}}
\authorrunning{}
% (feature abused for this document to repeat the title also on left hand pages)

% the affiliations are given next; don't give your e-mail address
% unless you accept that it will be published
\institute{CISPA Helmholtz Center for Information Security, Saarbr\"ucken, Germany\and University of California at Berkeley, Berkeley, USA}

%
% NB: a more complex sample for affiliations and the mapping to the
% corresponding authors can be found in the file "llncs.dem"
% (search for the string "\mainmatter" where a contribution starts).
% "llncs.dem" accompanies the document class "llncs.cls".
%

\maketitle

\begin{abstract}
  Hyperproperties are properties that describe the
  correctness of a system as a relation between multiple executions. Hyperproperties generalize trace properties and include information-flow security requirements, like noninterference, as well as requirements like symmetry, partial observation, robustness, and fault tolerance.
  We initiate the study of the specification and verification of hyperproperties of Markov decision processes (MDPs). We introduce the temporal logic \emph{\phl (Probabilistic Hyper Logic)}, which extends classic probabilistic logics with quantification over schedulers and traces. \phl can express a wide range of hyperproperties for probabilistic systems, including both classical applications, such as probabilistic noninterference, and novel applications in areas such as robotics and planning. While the model checking problem for \phl is in general undecidable, we provide methods both for proving and for refuting formulas from a fragment of the logic. The fragment includes many probabilistic hyperproperties of interest.
\end{abstract}

\section{Introduction}\label{sec:intro}
Ten years ago, Clarkson and Schneider coined the term
\emph{hyperproperties}~\cite{Clarkson+Schneider/10/Hyperproperties} for the class of properties that
describe the correctness of a system as a relation between multiple
executions. Hyperproperties include information-flow security
requirements, like noninterference~\cite{Goguen+Meseguer/1982/SecurityPoliciesAndSecurityModels}, as well as many other 
types of system requirements that cannot be expressed as trace properties,
including symmetry, partial observation, robustness, and fault
tolerance. Over the past decade, a rich set of
tools for the specification and verification of hyperproperties have been developed.
\hyperltl and \hyperctl~\cite{Clarkson+Finkbeiner+Koleini+Micinski+Rabe+Sanchez/2014/TemporalLogicsForHyperproperties} are extensions to LTL and CTL$^*$ that can
express a wide range of hyperproperties.  There are a number of algorithms and
tools for hardware model checking~\cite{frs15,DBLP:conf/cav/CoenenFST19}, satisfiability checking~\cite{fhs17},
and reactive synthesis~\cite{fhlst18} for hyperproperties.

The natural next step is to consider probabilistic
systems. Randomization plays a key role in the design of
security-critical and distributed systems. In fact, randomization is
often added specifically to implement a certain hyperproperty.  For
example, randomized mutual exclusion protocols use a coin flip to
decide which process gets access to the critical resource in order to
avoid breaking the symmetry based on the process id~\cite{DBLP:conf/ifm/Baier10}. 
Databases employ privacy mechanisms based on randomizaton 
in order to guarantee (differential) privacy~\cite{DifferentialPrivacy}.

Previous work on probabilistic hyperproperties~\cite{DBLP:conf/qest/AbrahamB18} has focussed on the specification and verification of probabilistic hyperproperties for
Markov chains. The logic HyperPCTL~\cite{DBLP:conf/qest/AbrahamB18} extends the standard probabilistic logic PCTL with quantification over states. For example, the HyperPCTL formula 
\[
\forall s.\forall s'.\, (\mathit{init}_s \wedge \mathit{init}_{s'}) \rightarrow \prob(\LTLdiamond \mathit{terminate}_s) = \prob(\LTLdiamond \mathit{terminate}_{s'})
\]
specifies that the probability that the system terminates is the same from all initial states. If the initial state encodes some secret, then the property guarantees that this secret is not revealed through the probability of termination.

Because Markov chains lack nondeterministic choice, they are a limited modeling tool. In an open system, the secret would likely be provided by an external environment, whose decisions would need to be represented by nondeterminism. In every step of the computation, such an environment would typically set the values of some low-security and some high-security input variables. In such a case, we would like to specify that the publicly observable behavior of our system does not depend on the infinite sequence of the values of the high-security input variables. Similarly, nondeterminism is needed to model the possible strategic decisions in autonomous systems, such as robots, or the content of the database in a privacy-critical system.

In this paper, we initiate the study of hyperproperties for \emph{Markov decision
processes} (MDPs). To formalize hyperproperties in this setting, we introduce \phl, a general temporal logic for
probabilistic hyperproperties. The nondeterministic choices of an MDP are
resolved by a \emph{scheduler}\footnote{In the literature, schedulers are also known as strategies or policies.}; correspondingly, our logic
quantifies over schedulers. For example, in the \phl
formula
\[
\forall \sigma.\forall\sigma'.\, \prob(\LTLdiamond \mathit{terminate}_\sigma) = \prob(\LTLdiamond \mathit{terminate}_{\sigma'})
\]
the variables $\sigma$ and $\sigma'$ refer to schedulers. The formula specifies that the probability of termination is the same for all of the possible (infinite) combinations of the nondeterministic choices. If we wish to distinguish different types of inputs, for example those that are provided through a high-security variable $h$ vs. those provided through a low-security variable $l$, then the quantification can be restricted to those schedulers that make the same low-security choices:
\[
\forall \sigma.\forall\sigma'.\, (\forall \pi:\sigma. \forall \pi':\sigma'.\ \LTLsquare (l_\pi\leftrightarrow l_{\pi'})) \rightarrow  \prob(\LTLdiamond \mathit{terminate}_\sigma) = \prob(\LTLdiamond \mathit{terminate}_{\sigma'})
\]
The path quantifier $\forall \pi:\sigma$ works analogously to the
quantifiers in \hyperctl, here restricted to the paths of the Markov
chain induced by the scheduler assigned to variable~$\sigma$. The formula thus states
that all schedulers that agree on the low-security inputs induce
the same probability of termination.

As we show in the paper, \phl is a very expressive logic, thanks to the combination of scheduler quantifiers, path quantifiers and a probabilistic operator. \phl has both
classical applications, such as differential privacy, as well as novel
applications in areas such as robotics and planning. 
For example, we can quantify the interference of the plans of different agents in a multi-agent system, such as the robots in a warehouse, or we can specify the existence of an approximately optimal policy that meets given constraints.  A
consequence of the generality of the logic is that it is impossible to
simply reduce the model checking problem to that of a simpler temporal
logic in the style of the reduction of HyperPCTL to
PCTL~\cite{DBLP:conf/qest/AbrahamB18}. In fact, we show that the
emptiness problem for probabilistic B\"uchi automata (PBA) can be
encoded in \phl, which implies that the model checking problem for \phl is, in
general, undecidable.

We present two verification procedures that approximate the model
checking problem from two sides. The first
algorithm \emph{overapproximates} the model checking problem by
quantifying over a combined monolithic scheduler rather than a tuple of
independent schedulers. Combined schedulers have access to more
information than individual ones, meaning that the set of allowed
schedulers is overapproximated. This means that if a universal formula
is true for all combined schedulers it is also true for all tuples 
of independent schedulers.  The second procedure is a
bounded model checking algorithm that \emph{underapproximates} the model
checking problem by bounding the number of states of the
schedulers. This algorithm is obtained as a combination of a bounded
synthesis algorithm for hyperproperties, which generates the
schedulers, and a model checking algorithm for Markov chains, which
computes the probabilities on the Markov chains induced by the
schedulers. Together, the two algorithms thus provide methods both for proving
and for refuting a class of probabilistic hyperproperties for MDPs.

\paragraph{Related work} Probabilistic noninterference originated in information-flow security~\cite{InfFlowSecurity,InfFlowInteractivePrograms} and is a security policy that requires that the probability of every low trace should be the same for every low equivalent initial state. Volpano and Smith~\cite{VolpanoPNI} presented a type system for checking probabilistic noninterference of concurrent programs with probabilistic schedulers. Sabelfeld and Sands~\cite{SabelfeldPNI} defined a secure type system for multi-threaded programs with dynamic thread creation which improves on that of Volpano and Smith. None of these works is concerned with models combining probabilistic choice with nondeterminism, nor with general temporal logics for probabilistic hyperproperties.

The specification and verification of probabilistic hyperproperties have recently attracted significant attention. Abraham and Bonakdarpour~\cite{DBLP:conf/qest/AbrahamB18} are the first to study a temporal logic for probabilistic 
hyperproperties, called HyperPCTL. The logic allows for explicit quantification over the states of a Markov chain, and is capable of expressing information-flow properties like probabilistic noninterference. The authors present a model checking algorithm for verifying HyperPCTL on finite-state Markov chains. HyperPCTL was extended to a logic called HyperPCTL*~\cite{StatisticalMCHyperPCTL} that allows nesting of temporal and probabilistic operators, and a statistical model checking method for HyperPCTL* was proposed. 
Our present work, on the other hand is concerned with the specification and model checking of probabilistic hyperproperties for system models featuring both probabilistic choice and nondeterminism, which are beyond the scope of all previous  temporal logics for probabilistic hyperproperties. 
Probabilistic logics with quantification over schedulers have been studied in~\cite{StochasticGameLogic} and~\cite{PSL-ijcai2019}. However, these logics do not include quantifiers over paths.

Independently and concurrently to our work, probabilistic hyperproperties for MDPs were also studied in~\cite{abbd20} (also presented at ATVA'20). The authors extend \hyperpctl with quantifiers over schedulers, while our new
logic \phl extends \hyperctl with the probabilistic operator and quantifiers over schedulers.
Thus, \hyperpctl quantifies over states (i.e., the computation trees that start from the states), while \phl quantifies over paths.
Both papers show that the model checking problem is undecidable for the respective logics. The difference is in how the approaches deal with the undecidability result. 
For both logics, the problem is decidable when quantifiers are restricted to non-probabilistic memoryless schedulers.~\cite{abbd20} provides an SMT-based verification procedure for \mbox{\hyperpctl} for this class of schedulers. We consider general memoryful schedulers and present two methods for proving and for refuting formulas from a fragment of \phl.

\section{Preliminaries}\label{sec:preliminaries}
\subsection{Markov Decision Processes}\label{sec:MDP}

% MDP
\begin{definition}[Markov Decision Process (MDP)]\label{def:MDP}
A \emph{Markov Decision Process (MDP)} is a tuple $M = (S,\Act,\trans,\init,\AP,L)$ where 
$S$ is a finite set of states, 
$\Act$ is a finite set of actions, 
$\trans : S \times \Act \times S \to [0,1]$ is the transition probability function such that $\sum_{s'\in S}\trans(s,a,s') \in \{0,1\}$ for every $s \in S$ and $a \in \Act$, $\init : S \to [0,1]$ is the initial distribution such that $\sum_{s \in S} \init(s) = 1$, $\AP$ is a finite set of atomic propositions and $L: S \to 2^{\AP}$ is a labelling function. 
\end{definition}

\looseness=-1
% Paths
A finite path in an MDP $M =  (S,\Act,\trans,\init,\AP,L)$ is a sequence $s_0s_1\ldots s_n$ where for every $0 \leq i < n$ there exists $a_i \in \Act$ such that $\trans(s_i,a_i,s_{i+1})>0$. Infinite paths in $M$ are defined analogously. We denote with $\Pathsf(M)$ and $\Pathsi(M)$ the sets of finite and infinite paths in $M$. For an infinite path $\rho = s_0s_1\ldots$ and $i \in \nats$ we denote with $\rho[i,\infty)$ the infinite suffix $s_i s_{i+1}\ldots $. Given $s \in S$, define $\Pathsf(M,s) = \{s_0s_1\ldots s_n \in \Pathsf(M)\mid s_0=s\}$, and similarly $\Pathsi(M,s)$.
% MDP with initial state
We denote with $M_s = (S,\Act,\trans,\init_s,\AP,L)$ the MDP obtained from $M$ by making $s$ the single initial state, i.e., $\init_s(s) = 1$ and $\init_s(t) = 0$ for  $t \neq s$.

For a set $A$ we denote with $\dist(A)$ the set of probability distributions on $A$.
 
% Schedulers
\begin{definition}[Scheduler]\label{def:scheduler}
A \emph{scheduler} for an MDP $M = (S,\Act,\trans,\init,\AP,L)$ is a function 
$\scheduler : (S\cdot \Act)^* S \to \dist(\Act)$ such that for all sequences $s_0 a_0 \ldots  a_{n-1}s_n \in (S\cdot\Act)^* S$ it holds that if $\scheduler (s_0 a_0 \ldots  a_{n-1}s_n)(a) > 0$ then $\sum_{t\in S}\trans(s_n,a,t) > 0$, that is, each action in the support of $\scheduler (s_0 a_0 \ldots  a_{n-1}s_n)$ is enabled in $s_n$. 
We define $\Sched(M)$ to be the set consisting of all schedulers for an MDP $M$.
\end{definition}

% Induced MC
Given an MDP $M = (S,\Act,\trans,\init,\AP,L)$ and a scheduler $\scheduler$ for $M$, we denote with $M_{\scheduler}$ the \emph{Markov chain of $M$ induced by $\scheduler$}, which is defined as the tuple $M_{\scheduler} = ((S\cdot \Act)^* S,\trans_\scheduler,\init,AP,L_\scheduler)$ where for every sequence $h  = s_0 a_0 \ldots  a_{n-1}s_n \in (S\cdot \Act)^* S$ it holds that $\trans_\scheduler(h,h \cdot s_{n+1}) = \sum_{a \in \Act}\scheduler(h)(a)\cdot\trans(s_n,a, s_{n+1})$ and $L_\scheduler(h) = L(s_n)$. Note that $M_{\scheduler}$ is infinite even when $M$ is finite.
The different types  of paths in a Markov chain are defined as for MDPs.

% Finite memory schedulers
Of specific interest are  \emph{finite-memory} schedulers, which are schedulers that can be represented as finite-state machines. Formally, a finite-memory scheduler for $M$ is represented as a tuple $\mathcal T_\scheduler = (Q,\delta,q_0,\mathit{act})$, where $Q$ is a finite set of states, representing the memory of the scheduler, $\delta : Q \times S \times \Act \to Q$ is a memory update function, $q_0$ is the initial state of the memory, and $\mathit{act} : Q \times S \to \dist(\Act)$ is a function that based on the current memory state and the state of the MDP returns a distribution over actions. Such a representation defines a function $\scheduler : (S\cdot \Act)^* S \to \dist(\Act)$ as follows. First, let us define the function $\delta^* : Q \times (S\cdot \Act)^* \to Q$ as follows: $\delta^*(q,\epsilon) = q$ for all $q \in Q$, and $\delta^*(q,s_0 a_0 \ldots  s_na_ns_{n+1}a_{n+1}) = \delta(\delta^*(q,s_0 a_0 \ldots  s_na_n),s_{n+1},a_{n+1})$ for all $q \in Q$ and all $s_0 a_0 \ldots  s_na_ns_{n+1}a_{n+1} \in (S\cdot \Act)^*$. Now, we define the scheduler function represented by $\mathcal T_\scheduler$ by $\scheduler(s_0 a_0 \ldots  s_na_ns_{n+1}) = \mathit{act}(\delta^*(s_0 a_0 \ldots  s_na_n),s_{n+1})$.

Finite-memory schedulers induce finite Markov chains with simpler representation. A finite memory scheduler $\scheduler$ represented by $\mathcal T_\scheduler = (Q,\delta,q_0,\mathit{act})$ induces the Markov chain $M_{\scheduler} = (S\times Q,\trans_\scheduler,\init_\scheduler,AP,L_\scheduler)$ where $\trans_\scheduler((s,q),(s',q'))=\sum_{a \in\Act}\mathit{act}(q,s)(a)\cdot\trans(s,a,s')$ if $q' = \delta(q,s)$, otherwise $\trans_\scheduler((s,q),(s',q'))= 0$, and $\init_\scheduler(s,q) = \init(s)$ if $q = q_0$ and $\init_\scheduler(s,q) = 0$ otherwise.

% Deterministic schedulers
A scheduler $\scheduler$ is \emph{deterministic} if for every $h \in (S\cdot \Act)^* S$ it holds that $\scheduler(h)(a) = 1$ for exactly one $a \in \Act$. By abuse of notation, a deterministic scheduler can be represented as a function $\scheduler : S^+ \to \Act$, that maps a finite sequence of states to the single action in the support of the corresponding distribution. Note that for deterministic schedulers we omit the actions from the history as they are uniquely determined by the sequence of states. We write $\DSched(M)$ for the set of deterministic schedulers for the MDP $M$.

\subsection{Probability Spaces}\label{sec:probabilities}

A \emph{probability space} is a triple $(\Omega,\mathcal F,\Prob)$, where $\Omega$ is a sample space, $\mathcal F \subseteq 2^{\Omega}$ is a $\sigma$-algebra and $\Prob : \mathcal F \to [0,1]$ is a probability measure.

Given a Markov chain $C = (S,\trans,\init,AP,L)$, it is well known how to associate a probability space $(\Omega^C,\mathcal F^C,\Prob^C)$ with $C$. The sample space $\Omega^C = \Pathsi(C)$ is the set of infinite paths in $C$, where the sets of finite and infinite paths for a Markov chain are defined in the same way as for MDP. The $\sigma$-algebra $\mathcal F^C$ is  the smallest $\sigma$-algebra that for each $\pi \in \Pathsf(C)$ contains the set $\cyl_C(\pi) = \{\rho \in \Pathsi(C) \mid \exists \rho' \in \Pathsi(C) : \rho = \pi \cdot \rho'\}$ called the cylinder set of the finite path $\pi$. $\Prob^C$ is the unique probability measure such that for each $\pi = s_0\ldots s_n \in \Pathsf(C)$ it holds that $\Prob^C(\cyl(\pi)) = \init(s_0) \cdot \prod_{i=0}^{n-1} \trans(s_i,s_{i+1})$.

Analogously, given any state $s \in S$ we denote with $(\Omega^C,\mathcal F^C,\Prob^C_s)$ the probability space for paths in $C$ originating in the state $s$, i.e., the probability space associated with the Markov chain $C_s$ (where $C_s$ is defined as for MDPs).

When considering a Markov chain $M_\scheduler$ induced by an MDP $M$ and a scheduler $\scheduler$, we write $\Prob_{M,\scheduler}$ and $\Prob_{M,\scheduler,s}$ for the sake of readability. 

\subsection{Linear Temporal Properties and Omega-Automata}\label{sec:automata}

\emph{LTL (Linear-time Temporal Logic)}~\cite{LTL} is a logic commonly used for specifying linear-time properties of reactive systems. In addition to Boolean connectives it contains the usual temporal operators Next $\LTLcircle$ and Until $\LTLuntil $, and the derived operators Weak Until $\LTLweakuntil$, Eventually $\LTLeventually$ and Globally~$\LTLglobally$. LTL formulas are defined over a set of atomic propositions $\AP$ and interpreted over infinite sequences over $2^{\AP}$ (see~\cite{LTL} for the semantics of LTL). We denote the satisfaction of an LTL formula $\varphi$ by an infinite sequence $w \in (2^{AP})^\omega $ by $w \models \varphi$. 

Following the convention of ~\cite{PrinciplesOfModelChecking} we use LTL notation for describing events: for example if $G$ is a set of states in a Markov chain $C$ we write $\LTLglobally\LTLfinally G$ for the event that $G$ is visited infinitely often, and write $\pi \models \LTLglobally\LTLfinally G$ instead of $\pi \in \LTLglobally\LTLfinally G$.

A \emph{deterministic Rabin automaton} over a given finite alphabet $\Lambda$ is a tuple
$\mathcal A = (Q,\Lambda,\delta,q_0,(B_{j},G_{j})_{j=1}^{m})$, where
$Q$ is finite set of states, 
$\delta : Q \times \Lambda \to Q$ is the transition function, 
$q_0 \in Q$ is the initial state and the pairs $(B_{j},G_{j})$ of sets of states define the accepting condition of $\mathcal A$. A run of $\mathcal{A}$ on an infinite word $w = \alpha_0\alpha_1\dots \in \Lambda^\omega$ is an infinite sequence $q_0 q_1 q_2\ldots\in Q^\omega$ of states such that $q_0$ is the initial state, and for every $i \in \mathbb{N}$ it holds that $q_{i+1} = \delta(q_i,\alpha_{i})$. Without loss of generality we assume that the transition function is total, and hence there exists exactly one run of $\mathcal A$ on each $w$ since $\mathcal A$ is deterministic. A run $q_0q_1q_2\ldots$ is accepting if it satisfies the Rabin condition which requires that for some pair $(B_j,G_j)$ the set $B_j$ is visited finitely often and the set $G_j$ is visited infinitely often. An infinite word $w$ is accepted by $\mathcal A$ if there exists an accepting run of $\mathcal A$ on $w$. We denote the set of infinite words accepted by 
$\mathcal{A}$ is, called the language of $\mathcal A$, by $\mathcal L(\mathcal{A})$. 
We define the size $|\mathcal A|$ to be the number of its states, i.e., $|\mathcal A| = |Q|$.
It is well know that for every LTL formula $\varphi$ one can construct a deterministic Rabin Automaton $\mathcal A_\varphi$ with $|\mathcal A_\varphi| = 2^{2^{\mathcal O(|\varphi|)}}$ such that $\mathcal L(\mathcal A) = \{w  \in 2^{\AP} \mid w \models \varphi\}$.

If $\pi=s_0s_1\ldots$ is an infinite path in an MDP or a Markov chain $M$, and $\varphi$ an LTL formula, we write $\pi \models_M \varphi$ meaning that $L(s_0)L(s_1)\models \varphi$.

A \emph{safety property} is a linear time property such that every word that does not satisfy the property has a finite prefix all of whose extensions violate the property. Safety properties can be represented by safety automata.
A \emph{deterministic safety automaton} over an alphabet $\Lambda$ is a tuple
$\mathcal A = (Q,\Lambda,\delta,q_0)$ where the transition function $\delta$ is partial and every infinite run is accepting.

\section{The Logic \phl}\label{sec:logic}
In this section we define the syntax and semantics of \phl, the logic which we introduce and study in this work. \phl allows for quantification over schedulers and integrates features of temporal logics for hyper properties, such as \hyperltl and \hyperctl~\cite{Clarkson+Finkbeiner+Koleini+Micinski+Rabe+Sanchez/2014/TemporalLogicsForHyperproperties}, and probabilistic temporal logics such as PCTL*.

\subsection{Examples of \phl Specifications}\label{sec:examples-logic}

We illustrate the expressiveness of PHL with several applications beyond those in information-flow security, from the domains of robotics and planning.

\begin{example}[Action cause]\label{ex:cause}
Consider the question whether a car on a highway that enters the opposite lane (action $b$) when there is a car approaching from the opposite direction (condition $p$) increases the probability of an accident (effect $e$). This can be formalized as the property stating that there exist two deterministic schedulers $\sigma_1$  and $\sigma_2$ such that (i) in $\sigma_1$ the action $b$ is never taken when $p$ is satisfied, (ii) the only differences between $\sigma_1$ and $\sigma_2$ can  happen when $\sigma_2$ takes action $b$ when $p$ is satisfied, and (iii) the probability of $e$ being true eventually is higher in the Markov chain induced by $\sigma_2$ than in the one for  $\sigma_1$. To express this property in our logic, we will use \emph{scheduler quantifiers} quantifying over the schedulers for the MDP. To capture the condition on the way the schedulers differ, we will use \emph{path quantifiers} quantifying over the paths in the Markov chain induced by each scheduler. The atomic propositions in a \phl formula are indexed with path variables when they are interpreted on a given path, and with scheduler variables when they are interpreted in the Markov chain induced by that scheduler. Formally, we can express the property with the \phl  formula 
\[
\exists\sigma_1\exists\sigma_2.\;
(\forall\pi_1:{\sigma_1}\forall\pi_2:{\sigma_2}.\; 
(\LTLglobally 
\neg(p_{\pi_{1}} \wedge \LTLnext b_{\pi_1})) \wedge\psi) \wedge
\prob(\LTLeventually e_{\sigma_1}) < \prob( \LTLeventually e_{\sigma_2}),
\]
where 
$\psi = \big((\bigwedge _{a \in \Act} (\LTLnext a_{\pi_1}  \leftrightarrow \LTLnext a_{\pi_2}))
\vee (p_{\pi_2} \wedge \LTLnext b_{\pi_2})\big) \LTLweakuntil (\bigvee_{q \in \AP\setminus\Act} (q_{\pi_1} \not\leftrightarrow q_{\pi_2}))$.

The two conjuncts of $\forall\pi_1:{\sigma_1}\forall\pi_2:{\sigma_2}.\; 
(\LTLglobally  \neg(p_{\pi_{1}} \wedge \LTLnext b_{\pi_1})) \wedge\psi$ capture conditions (i) and (ii) above respectively, and $\prob(\LTLeventually e_{\sigma_1}) < \prob( \LTLeventually e_{\sigma_2})$ formalizes (iii). Here we assume that actions are represented in $\AP$, i.e., $\Act \subseteq \AP$ \qed
\end{example}

\begin{example}[Plan non-interference]\label{ex:robot}
Consider two robots in a warehouse, possibly attempting to reach the same location. Our goal is to determine whether all plans for the first robot to move towards the goal are robust against interferences from arbitrary plans of the other robot. That is, we want to  check whether for every plan of robot 1 the probability that it reaches the goal under an arbitrary plan of robot 2 is close to that of the same plan for robot 1 executed under any other plan for robot 2.
 We can express this property in \phl by using \emph{quantifiers over schedulers} to quantify over the joint deterministic plans of the robots, and using \emph{path quantifiers}  to express the condition that in both joint plans robot 1 behaves the same. 
Formally, we can express the property with the \phl  formula 
\[\begin{array}{l}
\forall\sigma_1\forall\sigma_2.\;
(\forall\pi_1:{\sigma_1}\forall\pi_2:{\sigma_2}.\; 
\LTLglobally (
\mathit{move1}_{\pi_{1}} \leftrightarrow 
\mathit{move1}_{\pi_{2}})) \rightarrow
\\\phantom{\forall\sigma_1\forall\sigma_2. }\;
\prob(
\LTLeventually 
(\mathit{goal1}_{\sigma_1}\wedge 
\neg \mathit{goal2}_{\sigma_1})) - 
\prob(\LTLeventually 
(\mathit{goal1}_{\sigma_2}\wedge 
\neg \mathit{goal2}_{\sigma_2}))
 \leq \varepsilon,
\end{array}
\]
where $\sigma_1$ and $\sigma_2$  are scheduler variables, $\pi_1$ is a path variable associated with the scheduler for $\sigma_1$, and $\pi_2$ is a path variable associated with the scheduler for $\sigma_2$. The condition $\forall\pi_1:{\sigma_1}\forall\pi_2:{\sigma_2}.\; 
\LTLglobally (
\mathit{move1}_{\pi_{1}} \leftrightarrow 
\mathit{move1}_{\pi_{2}})$ states that in both joint plans robot 1 executes the same moves, where the proposition $\mathit{move1}$ corresponds to robot 1 making a move towards the goal. The formula $\prob(
\LTLeventually 
(\mathit{goal1}_{\sigma_1}\wedge 
\neg \mathit{goal2}_{\sigma_1})) - 
\prob(\LTLeventually 
(\mathit{goal1}_{\sigma_2}\wedge 
\neg \mathit{goal2}_{\sigma_2}))
 \leq \varepsilon$ states that the difference in the probability of robot 1 reaching the goal under scheduler $\sigma_1$ and the probability of it reaching the goal under scheduler $\sigma_2$ does not exceed $\varepsilon$.
 \qed
 \end{example}

\begin{example}[Bounding performance loss]\label{ex:bounded-performance}
\phl allows us to compare the performance of different schedulers in an MDP. For instance we can specify the property that the difference in the probability of success of two schedulers that differ in certain action choices is bounded by a given constant.
For instance, consider an MDP where actions come in pairs, for example one corresponding to using a more precise (but more power consuming) sensor and the other one corresponding to using a less precise (but less power consuming) sensing mechanism. The property that we want to state is that for every policy replacing some of the more precise actions with their less precise counterpart has a bounded effect on the probability of reaching some set of goal states. 

This property is specified by the following formula where the propositions $a'$ and $a''$ correspond to the more precise and less precise actions in a pair:
\[\forall\sigma_1\forall\sigma_2.
\hform_{actions} \rightarrow 
\prob(\LTLeventually 
\mathit{goal_{\sigma_1}}) - 
\prob(\LTLeventually 
\mathit{goal_{\sigma_2}}) \leq \varepsilon, \text{ where}\] 
 $\hform_{actions} = 
\forall\pi_1:{\sigma_1}.\forall\pi_2 :{\sigma_2}.\big(\bigwedge_{a} 
(a'_{\pi_{1}}  \wedge a''_{\pi_{2}} \vee
 a''_{\pi_{1}}  \wedge a''_{\pi_{2}})
 \big)\LTLweakuntil 
\neg(\mathit{in}_{\pi_{1}}\leftrightarrow 
\mathit{in}_{\pi_{2}})$.
The formula $\hform_{actions}$ states that  both policies can differ only in replacing $a'$ in policy $\sigma_1$ by $a''$ in policy $\sigma_2$. The formula $\prob(\LTLeventually 
\mathit{goal_{\sigma_1}}) - 
\prob(\LTLeventually 
\mathit{goal_{\sigma_2}}) \leq \varepsilon$ states that the resulting decrease in the probability of reaching the goal is bounded by $\varepsilon$.
\qed
\end{example}

\smallskip
\begin{example}[Existence of near-optimal scheduler under hard constraint]\label{ex:hard-soft}
Consider an MDP $M= (S,\Act,\trans,\init,\AP,L)$ and two LTL specifications $\varphi_{hard}$ and $\varphi_{soft}$ over $\AP$. We can express in \phl the property that there exists scheduler $\scheduler$ for $M$ such that $\Prob_{M,\scheduler}(\{\pi \models \varphi_{hard}\})\geq c$ such that for every $\scheduler'\in\Sched(M)$ it holds that $\Prob_{M,\scheduler'}(\{\pi \models \varphi_{\mathit{soft}}\}) - \Prob_{M,\scheduler}(\{\pi \models \varphi_{\mathit{soft}}\})\leq \varepsilon$. That is, $\scheduler$ satisfies $\varphi_{hard}$ with probability at least $c$ and satisfies  $\varphi_{\mathit{soft}}$ with probability that is at most $\varepsilon$-away from that of the optimal scheduler. The formula is:
 \[\exists \sigma. \prob(\varphi_{hard}[\AP_{\sigma}/\AP]) \geq c \wedge \forall\sigma'.\prob(\varphi_{\mathit{soft}}[\AP_{\sigma'}/\AP]) - \prob(\varphi_{\mathit{soft}}[\AP_{\sigma}/\AP]) \leq \varepsilon,\]
 where $\varphi_{hard}[\AP_{\sigma}/\AP]$ is the formula $\varphi_{hard}$ in which each $a \in \AP$ is replaced by $a_\sigma$ and similarly for the other two formulas.
 \qed
\end{example}

\begin{example}[Differential privacy]\label{ex:diff-privacy}
Using \phl we can express the property of differential privacy~\cite{DifferentialPrivacy}. Intuitively, a randomized algorithm that produces responses to queries to an input database is differentially private if it behaves similarly on similar input databases.
The schedulers in the MDP encode all possible unbounded length sequences of updates and queries to a database. Update actions are deterministic and modify the database, and query actions are probabilistic and encode the randomized responses to queries by the privacy mechanism.
Suppose that an MDP $M_{dp} =  (S_{dp},\Act_{dp},\trans_{dp},\init_{dp},\AP_{dp},L_{dp})$ encodes some differential-privacy mechanism in the following way.

\begin{itemize}
\item Each state in $S_{dp}$ encodes the current contents of the database, together with the last update or query and the last given response, if any.
\item The set of actions $\Act_{dp}$ consists of two types of actions $\Act_{dp} =  \Act_{update} \uplus   \Act_{query}$, modelling the possible updates and queries respectively.
\item The transition probability function $\trans_{dp}$ is defined based on the privacy mechanism being modelled. The transitions for actions in $\Act_{update}$ are deterministic and encode how the database contents is modified. The transitions for actions $\Act_{query}$ are probabilistic and correspond to the randomized responses to queries as defined by the privacy mechanism. 
\item The atomic propositions in $\AP_{dp} = \AP_{update} \uplus \AP_{query} \uplus \AP_{resp}$ contain propositions for the action labels in order to allow comparison of the sequence of updates and queries, as well as labels describing the responses to the queries. The labelling function $L_{dp}$ maps each state to the atomic propositions reflecting the last action and query response stored in the state.
\end{itemize}

The \phl formula $\qform_{dp}$ given below encodes $(\varepsilon,\delta)$-differential privacy of the privacy mechanism encoded by an MDP of the above form. It states that for every two schedulers that correspond to sequences of updates that differ in at most one position, and that issue the same set of queries, the probabilities of each response to each query are "close" as defined by $\varepsilon$ and $\delta$.
\[\qform_{dp} = 
\forall\sigma_1\forall\sigma_2.
(\forall\pi_1:{\sigma_1}.\forall\pi_2:{\sigma_2}.\;
\varphi_{sim} \wedge \varphi_{queries} )
\rightarrow \varphi_{prob}\]
\begin{itemize}

\item $\varphi_{sim}$ specifies  that the two sequences differ in at most one update:
\[\varphi_{sim} = \LTLglobally 
\Big(
\Big(
\bigvee_{a \in \AP_{update}}
\neg(a_{\pi_{1}} \leftrightarrow a_{\pi_{2}} )\Big)
\rightarrow
\LTLnext\LTLglobally 
\neg \Big(\bigvee_{a \in \AP_{update}}
\neg(a_{\pi_{1}} \leftrightarrow a_{\pi_{2}} )\Big)\Big).
\]

\item $\varphi_{queries}$ specifies that the same set of queries are being issued:
\[\varphi_{queries} = \bigwedge_{a\in \AP_{query}} (\LTLeventually {a}_{\pi_{1}}\leftrightarrow \LTLeventually{a}_{\pi_{2}}).\]

\item $\varphi_{prob}$ specifies that for every query each response has  close enough probability of occurrence in the two schedulers (i.e., the two databases):
\[\varphi_{prob} = 
\bigwedge_{q \in \AP_{query}}
\bigwedge_{r \in \AP_{resp}} 
\prob(\LTLglobally({q}_{\sigma_1} \rightarrow 
{r}_{\sigma_1})) - 
\exp(\varepsilon)
\prob(\LTLglobally({q}_{\sigma_2} \rightarrow 
{r}_{\sigma_2})) < \delta.\]
\end{itemize}
\qed
\end{example}

\subsection{Syntax}\label{sec:syntax}
As we are concerned with hyperproperties interpreted over MDPs, our logic allows for quantification over schedulers and quantification over paths. 

% shceduler and path variables
To this end, let $\varss$ be a countably infinite set of \emph{scheduler variables} and 
let $\varsp$ be a countably infinite set of \emph{path variables}. According to the semantics of our logic, quantification over path variables ranges over the paths in a Markov chain associated with the scheduler represented by a given scheduler variable. To express this dependency we will associate path variables with the corresponding scheduler variable, writing $\pi:\sigma$ for a path variable $\pi$ associated with a scheduler variable $\sigma$. 
The precise use and meaning  of this notation will become clear below, once we define the syntax and semantics of the logic.  

% indexed atomic propositions
Given a set $\AP$ of atomic propositions, \phl formulas over $\AP$ will use atomic propositions indexed with scheduler variables, and both path and associated scheduler variables.
We define the sets of atomic propositions indexed with scheduler variables as
$\AP_{\varss} = \{ a_\sigma \mid a \in \AP,\sigma \in \varss\}$ and 
indexed with path variables as
$\AP_{\varsp} = \{ a_\pi \mid a \in \AP,\pi \in \varsp\}$.

% phl grammar   
\emph{\phl (Probabilistic Hyper Logic) formulas} are defined by the grammar
\[
\qform ::= \;\;
\forall \sigma.\; \qform \;\;\mid\;\; 
\qform \wedge \qform \;\;\mid\;\; 
\neg \qform \;\;\mid\;\; 
\hform  \;\;\mid\;\; 
\pexpr \bowtie c 
\]
where $\sigma \in \varss$ is a scheduler variable, $\hform$ is a \hyperctl formula, $\pexpr$ is a \emph{probabilistic expression} defined below, $\bowtie \in \{\le,\leq,\ge,\geq\}$, and $c \in \mathbb{Q}$.

Formulas in \hyperctl, introduced in~\cite{Clarkson+Finkbeiner+Koleini+Micinski+Rabe+Sanchez/2014/TemporalLogicsForHyperproperties}, are constructed by the grammar
\[
\hform ::= 
a_\pi \;\mid\; 
\hform \wedge \hform \;\mid\; 
\neg \hform \;\mid\; 
\LTLnext \hform  \;\mid\; 
\hform \LTLuntil \hform \;\mid\; 
\forall \pi:\sigma.\; \hform 
\]
where $\pi$ is a path variable associated with a scheduler variable $\sigma$, and $a\in\AP$.

\emph{Probability expressions} are defined by the grammar
\[
\pexpr ::= 
\prob(\pathform) \;\mid\; 
\pexpr + \pexpr \;\mid\; 
c \cdot \pexpr
\]
where $\prob$ is the \emph{probabilistic operator}, $c \in \mathbb{Q}$, and $\varphi$ is an LTL formula~\cite{LTL} defined by the grammar below, where $a \in \AP$ and $\sigma$ is a scheduler variable.
\[
\pathform ::= 
a_{\sigma} \;\mid\; 
\pathform \wedge \pathform \;\mid\; 
\neg \pathform \;\mid\; 
\LTLnext \pathform  \;\mid\; 
\pathform \LTLuntil \pathform.
%\;\mid\;\pexpr \bowtie c
\]
We call formulas of the form $\pexpr \bowtie c$ \emph{probabilistic predicates}.

A \phl formula $\qform$ is \emph{well-formed} if each path quantifier for $\pi:\sigma$ that appears in $\qform$ is in the scope of a scheduler quantifier with the scheduler variable $\sigma$.

A \phl formula is \emph{closed} if all occurrences of scheduler and path variables are bound by scheduler and path quantifiers respectively.

In the following we consider only closed and well-formed \phl formulas.

\paragraph{Discussion} Intuitively, a \phl formula is a Boolean combination of formulas consisting of a scheduler quantifier prefix followed by a formula without scheduler quantifiers constructed from probabilistic predicates and \hyperctl formulas via propositional operators. Thus, interleaving path quantifiers and probabilistic predicates is not allowed in \phl. This design decision is in line with the fact that probabilistic temporal logics like PCTL$^*$ replace the path quantifiers with the probabilistic operator that can be seen as their quantitative counterpart.
We further chose to not allow nesting of probabilistic predicates and temporal operators, as  in all the examples that we considered we never encountered the need for nested $\prob$ operators. Moreover, allowing arbitrary nesting of probabilistic and temporal operators would immediately make the model checking problem for the resulting logic undecidable, following from the results in~\cite{StochasticGamesBranchingTimeObjectives}. 

\subsection{Self-Composition for MDPs}\label{sec:self-composition}

In order to define the semantics of \phl we first introduce the self-composition operation for MDPs, which lifts to MDPs the well-known self-composition of transition systems that is often used in the model checking of hyperproperties.

Let us fix, for the reminder of the section, an MDP $M =  (S,\Act,\trans,\init,\AP,L)$.

% self-composition
\begin{definition}[$n$-self-composition of an MDP]\label{def:self-composition}
Let $M =  (S,\Act,\trans,\init,\AP,L)$ be an MDP and $n \in \nats_{>0}$ be a constant. 
The \emph{$n$-self-composition of $M$} is the MDP 
$M^n =  (S^n,\Act^n,\widehat\trans,\widehat\init,\AP,\widehat L)$ 
with the following components.
\begin{itemize}
\item $S^n = \{(s_1,\ldots,s_n) \mid s_i \in S \text{ for all } 1\leq i \leq n\}$ is the set of states.
\item $\Act^n = \{(a_1,\ldots,a_n) \mid a_i \in \Act \text{ for all } 1\leq i \leq n\}$ is the set of actions.
\item The transition probability function $\widehat\trans$ is defined such that 
for every pair of states $(s_1,\ldots,s_n),(s_1',\ldots,s_n')\in S^n$ and every action $(a_1,\ldots,a_n) \in \Act^n$:
\[
\widehat\trans((s_1,\ldots,s_n),(a_1,\ldots,a_n),(s_1',\ldots,s_n')) = 
\prod_{i=1}^n\trans(s_i,a_i,s'_i).
\]
\item We define the initial distribution such that 
$\widehat\init((s_1,\ldots, s_n)) = \init(s)$  if $s_1 = \ldots = s_n = s$  for some $s\in S$ and $\widehat\init((s_1,\ldots, s_n)) = 0$ otherwise.
\item The labelling function $\widehat L : S^n \to (2^\AP)^n$ maps the states of $M^n$ to $n$-tuples of subsets of $\AP$ (in contrast to definition Definition~\ref{def:MDP} where states are mapped to subsets of $\AP$) and is given by $\widehat L((s_1,\ldots,s_n)) = (L(s_1),\ldots,L(s_n))$. 
\end{itemize}
\end{definition}

% Markov chain for self-composed MDP
Naturally, a scheduler $\widehat\scheduler \in \Sched(M^n)$ induces a Markov chain $M^n_{\widehat{\scheduler}}$.

% scheduler composition
Given schedulers $\scheduler_1,\ldots,\scheduler_n \in \Sched(M)$, their \emph{composition}, a scheduler $\overline \scheduler : (S^n\cdot \Act^n)^* S^n \to \dist(\Act^n)$ for $M^n$, is denoted $\overline \scheduler= \scheduler_1\parallel\ldots\parallel\scheduler_n$ and 
such that for every $\overline h = (s_{1,1},\ldots,s_{1,n})(a_{1,1},\ldots,a_{1,n})\ldots(s_{k,1},\ldots,s_{k,n}) \in (S^n\cdot \Act^n)^* S^n$
and $\overline a = (a_{k+1,1},\ldots,a_{k+1,n}) \in \Act^n$,
$
\overline\scheduler(\overline h)(\overline a) =
\prod_{i=1}^n \scheduler_i(s_{1,i}a_{1,i}\ldots s_{k,i})(a_{k+1,i}).
$

\subsection{Scheduler and Path Assignments}\label{sec:assignments}

Let $\varss$ and $\varsp$ be the sets of scheduler and path variables respectively.

% scheduler assignment
A \emph{scheduler assignment} is a vector of pairs $\sassign \in \bigcup_{n \in \nats}(\varss \times \Sched(M))^n $ that assigns schedulers to some of the scheduler variables. 
Given a scheduler assignment $\sassign = ((\sigma_1,\scheduler_1),\ldots,(\sigma_n,\scheduler_n))$, we denote by $|\sassign|$ the length (number of pairs) of the vector. For a scheduler variable $\sigma \in \varss$ we define $\sassign(\sigma) = \scheduler_i$ where $i$ is the maximal index such that $\sigma_i = \sigma$. If such an index $i$ does not exits, $\sassign(\sigma)$ is undefined.
For a scheduler assignment $\sassign = ((\sigma_1,\scheduler_1),\ldots,(\sigma_n,\scheduler_n))$, a scheduler variable $\sigma \in \varss$, and a scheduler $\scheduler \in \Sched(M)$ we define the scheduler assignment $\sassign[\sigma \mapsto \scheduler] = ((\sigma_1,\scheduler_1),\ldots,(\sigma_n,\scheduler_n),(\sigma,\scheduler))$ obtained by adding the pair $(\sigma,\scheduler)$ to the end of the vector $\sassign$.

% Markov chain induced by scheduler assignment
Given the MDP $M$, let $\sassign = ((\sigma_1,\scheduler_1),\ldots,(\sigma_n,\scheduler_n))$ be a scheduler assignment, and consider $M^{|\sassign|}$, the $|\sassign|$-self composition of $M$. $\sassign$ defines a scheduler for $M^{|\sassign|}$, which is the product of the schedulers in $\sassign$, i.e., $\overline{\scheduler} = \scheduler_1 \parallel \ldots \parallel \scheduler_n$. Let $M_{\sassign}$ be the Markov chain induced by $\overline{\scheduler}$. If $\widehat{s}$ is a state in $M_{\sassign}$, we denote by $M_{\sassign,\widehat s}$ the Markov chain obtained from $M_{\sassign}$ by making $\widehat s$ the single initial state.

% labeling function in product MDP/MC
Note that the labeling function $\widehat L$ in $M^{|\sassign|}$ maps the states in $S^{|\sassign|}$ to $|\sassign|$-tuples of sets of atomic predicates, that is $\widehat{L}(\widehat{s}) = (L_{1},\ldots,L_{|\sassign|})$. Given a scheduler variable $\sigma$ for which $\sassign(\sigma)$ is defined, we write $\widehat{L}(\widehat{s})(\sigma)$ for the set of atomic predicates $L_i$, where $i$ is the maximal position in $\sassign$ in which $\sigma$ appears. 

% path assignments
We define path assignments similarly to scheduler assignments. 
A \emph{path assignment} is a vector of pairs of path variables and paths in $\Pathsi(M)$. 
More precisely, a path assignment $\passign$ is an element of 
$\bigcup_{m \in \nats}(\varsp \times \Pathsi(M))^m$. Analogously to scheduler assignments, for a path variable $\pi$ and a path $\rho \in \Pathsi(M)$, we define $\passign(\pi)$ and $\passign[\pi \mapsto \rho]$. For  $\passign = ((\pi_1,\rho_1),\ldots,(\pi_n,\rho_n))$ and $j \in \nats$, we define $\passign[j,\infty] = ((\pi_1,\rho_1[j,\infty]),\ldots,(\pi_n,\rho_n[j,\infty]))$ to be the path assignment that assigns to each $\pi_i$ the suffix $\rho_i[j,\infty]$ of the path $\rho_i$.

\subsection{Semantics of \phl}\label{sec:semantics}
We are now ready to define the semantics of \phl formulas. Recall that we consider only closed and well-formed \phl formulas. \phl formulas are interpreted over an MDP and a scheduler assignment. The interpretation of \hyperctl formulas requires additionally a path assignment. Probabilistic expressions and LTL formulas are evaluated in the Markov chain for an MDP induced by a scheduler assignment. As usual, the satisfaction relations are denoted by $\models$.

For an MDP $M$ and a scheduler assignment $\sassign$ we define
\[
\begin{array}{lcl}
M,\sassign \models \forall \sigma. \qform & \quad\text{iff} & \quad
\text{for all } \scheduler \in \Sched(M):\; M,\sassign[\sigma\mapsto\scheduler] \models \qform; \\
M,\sassign \models \qform_1 \wedge \qform_2 & \quad\text{iff} & \quad
M,\sassign \models \qform_1 \text{ and } M,\sassign \models \qform_2; \\
M,\sassign \models \neg\qform & \quad\text{iff} & \quad  
M,\sassign \not\models \qform;\\
M,\sassign \models \hform & \quad\text{iff} & \quad  
M,\sassign,\passign_\emptyset \models \hform, \text{where } \passign_\emptyset \text{ is the empty path assignment};\\
M,\sassign \models \pexpr \bowtie c & \quad\text{iff} & \quad  
\llbracket \pexpr \rrbracket_{{M}_{\sassign}} \bowtie c. \\
\end{array}
\]

For an MDP $M$, scheduler assignment $\sassign$, and path assignment $\passign$ we define
\[
\begin{array}{lcl}
M,\sassign,\passign \models a_\pi & \quad\text{iff} & \quad  
a \in L(\Pi(\pi)[0]);\\
M,\sassign,\passign \models \hform_1 \wedge \hform_2 & \quad\text{iff} & \quad  
M,\sassign,\passign \models \hform_1 \text{ and } M,\sassign,\passign \models \hform_2;\\
M,\sassign,\passign \models \neg\hform & \quad\text{iff} & \quad  
M,\sassign,\passign \not\models \hform;\\
M,\sassign,\passign \models \LTLnext \hform & \quad\text{iff} & \quad  
M,\sassign,\passign[1,\infty] \models \hform;\\
M,\sassign,\passign \models \hform_1\LTLuntil \hform_2 & \quad\text{iff} & \quad  
\text{there exists } i \geq 0: M,\sassign,\passign[i,\infty] \models \hform_2 \text{ and}\\
&&\phantom{\text{there exists } i \geq 0:} \quad\text{ for all } j < i: M,\sassign,\passign[j,\infty] \models \hform_1;\\
M,\sassign,\passign \models \forall \pi:\sigma.\; \hform & \quad\text{iff} & \quad  
\text{ for all } \rho \in \Pathsi(C): M,\sassign,\passign[\pi \mapsto \rho] \models \hform,\\
\end{array}
\]
where in the last item $C$ is the Markov chain $M_{\sassign(\sigma)}$ when $\passign$ is the empty path assignment, and otherwise the Markov chain $M_{\sassign(\sigma),\passign(\pi')[0]}$ where $\pi'$ is the path variable associated with scheduler variable $\sigma$ that was most recently added to $\passign$.\looseness=-1

For Markov chain $C$ of the form $M_{\sassign}$ or $M_{\sassign,\widehat s}$, where $\sassign$ is a scheduler assignment and $\widehat s$ is a state in $M_\sassign$  the semantics $\llbracket \cdot \rrbracket_C$ of probabilistic  expressions is:\looseness=-1
\[
\begin{array}{lcl}
\llbracket \prob(\pathform) \rrbracket_{C} & \; = \; & 
\Prob^C(\{\rho \in \Pathsi(C) \mid C,\rho \models \pathform\}); \\
\llbracket \pexpr_1 + \pexpr_2 \rrbracket_{C} & \; = \; & 
\llbracket \pexpr_1 \rrbracket_{C} + \llbracket \pexpr_2 \rrbracket_{C}; \quad
\llbracket c \cdot \pexpr \rrbracket_{C}  \; = \; 
c \cdot \llbracket \pexpr \rrbracket_{C}, \\
\end{array}
\]
where the semantics of path formulas (i.e., LTL formulas) is defined by
\[
\begin{array}{lcl}
C,\rho \models a_{\sigma} & \quad\text{iff} & \quad  
a \in \widehat L(\rho[0])(\sigma);\\
C,\rho \models \pathform_1 \wedge \pathform_2 & \quad\text{iff} & \quad  
C,\rho \models \pathform_1 \text{ and } C,\rho \models \pathform_2; \\
C,\rho \models \neg\pathform & \quad\text{iff} & \quad  
C,\rho \not\models \pathform;\\
C,\rho \models \LTLnext \pathform & \quad\text{iff} & \quad  
C,\rho[1,\infty] \models \pathform;\\
C,\rho \models \pathform_1\LTLuntil \pathform_2 & \quad\text{iff} & \quad  
\text{there exists } i \geq 0: C,\rho[i,\infty] \models \pathform_2 \text{ and}\\
&&\phantom{\text{there exists } i \geq 0:} \quad\text{ for all } j < i: C,\rho[j,\infty] \models \pathform_1.\\
%C,\rho \models \pexpr \bowtie c & \quad\text{iff} & \quad  
%\llbracket \pexpr \rrbracket_{C_{\rho[0]}} \bowtie c.\\
\end{array}
\]

Note that $\Prob^C(\{\rho \in \Pathsi(C) \mid C,\rho \models \pathform\})$ is well-defined as it is a known fact~\cite{PrinciplesOfModelChecking} that the set $\{\rho \in \Pathsi(C) \mid C,\rho \models \pathform\}$ is measurable.

We say that an MDP $M$ \emph{satisfies} a closed well-formed \phl formula $\Phi$, denoted $M \models \Phi$ iff $M,\Sigma_{\emptyset}\models \Phi$, where $\Sigma_{\emptyset}$ is the empty scheduler assignment.

% deterministic vs randomized quantifiers

Since \phl includes both scheduler and path quantification, the sets of deterministic and randomized schedulers are not interchangeable with respect to the \phl semantics. That is, there exists an MDP $M$ and formula $\qform$ such that if quantifiers are interpreted over $\Sched(M)$, then $M \models \qform$, and if quantifiers are interpreted over $\DSched(M)$ then $M \not\models \qform$. See Appendix~\ref{sec:rand-sched} for an example.

\subsection{Undecidability of \phl Model Checking}\label{sec:undecidability}
 Due to the fact that \phl allows quantification over both schedulers and paths, the model checking problem for \phl is undecidable. The proof is based on a reduction from the emptiness problem for probabilistic B\"uchi automata (PBA), which is known to be undecidable~\cite{DecisionProblemsProbabilisticBuchiAutomata}.

\begin{theorem}\label{thm:undecidability-pba}
The model checking problem for \phl is undecidable.
\end{theorem}
\begin{proof}
We prove the undecidability of the model checking problem for \phl via a reduction from the emptiness problem for probabilistic B\"uchi automata (PBA), which is known to be undecidable~\cite{DecisionProblemsProbabilisticBuchiAutomata}. More precisely, we show how for each PBA $\mathcal B$ to construct an MDP $M$ and an \phl formula $\Phi$ such that there exists a infinite word $w$ over $\mathcal B$'s alphabet such that $\Prob_{\mathcal B} (w) >0$ iff $M \models \Phi$.

Let $\mathcal B = (Q,\Lambda,\delta,\mu,F)$ be a PBA, where $Q$ is a finite set of states, $\Lambda$ is a finite alphabet, $\delta : Q \times \Lambda \times Q \to [0,1]$ is the transition probability function, $\mu$ is the initial distribution and $F \subseteq Q$ is the set of accepting states.   

We will now define an MDP in which the scheduler picks letters in $\Lambda$ and the probabilistic choices mimic the behaviour of the PBA. The \phl formula will assert the existence of a scheduler such that in the resulting Markov chain \emph{all paths have the same scheduler choices} (i.e., the scheduler is "blind") and in the resulting Markov chain the probability of visiting an accepting state  infinitely often is greater than $0$. This is equivalent to the scheduler generating word $w \in \Lambda^\omega$ regardless of what  states are visited by the automaton, and $\Prob_{\mathcal B} (w) >0$.

We define the MDP $M = (S,\Act,\trans,\init,\AP,L)$ as follows.
\begin{itemize}
\item The states and actions are respectively $S =  Q \times (\Lambda \uplus \{\bot\})$ and $\Act = \Lambda$.
\item The transition probability function $\trans$ is defined by the function $\delta$ in $\mathcal B$ and is such that for every $(q,\alpha),(q',\alpha') \in S$ and $\beta \in \Lambda$ we have 
$\trans((q,\alpha),\beta,(q',\alpha')) = 
\delta(q,\beta,q')$  if $\alpha' = \beta$,
and $\trans((q,\alpha),\beta,(q',\alpha')) = 0$
otherwise. 
\item The initial distribution $\init$ is defined by the initial distribution $\mu$ from $\mathcal B$ and is such that $\init((q,\alpha)) = \mu(q)$ if $\alpha=\bot$ and $\init((q,\alpha)) = 0$ for $\alpha \neq \bot$.
\item The atomic propositions $\AP = \Lambda \uplus \{f\}$ correspond to the alphabet of $\mathcal B$ plus a fresh proposition $f$, and $L$ is such that for every $(q,\alpha)\in S$ we have that:
\begin{itemize}
\item if $\alpha=\bot$: if $q \not \in F$, then $L((q,\alpha)) = \emptyset$, and if $q \in F$, $L((q,\alpha)) = \{f\}$;
\item if $\alpha\in \Lambda$: if $q \not \in F$, then $L((q,\alpha)) = \{\alpha\}$, otherwise $L((q,\alpha)) = \{\alpha,f\}$.
\end{itemize} 
\end{itemize}

Let $
\qform = \exists \sigma.\; \big(
\exists \pi:{\sigma}.
\forall \pi':{\sigma}.\;
\LTLglobally \bigwedge_{\alpha \in \Lambda}
(\alpha_{\pi}\leftrightarrow 
\alpha_{\pi'})\big) \wedge
\prob(\LTLglobally\LTLfinally f_\sigma)>0.
$

\smallskip

We now prove that  for the MDP $M$ and the formula $\Phi$ defined above it holds that there exists a word $w$ such that $\Prob_{\mathcal{B}}(w) > 0$ if and only if $M \models \Phi$.

($\Rightarrow$) Let $w=\alpha_0\alpha_1 \ldots \in \Lambda^\omega$ be such that $\Prob_{\mathcal{B}}(w) > 0$.

We define the deterministic scheduler $\scheduler: (S \times \Act)^*S \to \Act$ such 
that for any $s_1a_1\ldots s_n a_n s'\in (S \times \Act)^*S$ we have 
$\scheduler(s_1a_1\ldots s_n a_n s') = \alpha_n$. That is, the scheduler $\scheduler$ blindly follows the word  $w$. Thus, it is clear that all paths in $M_\scheduler$ have the same sequence of labels (which is the word $w$), and hence, the first conjunct in $\qform$ will be satisfied. Since $\trans$ is defined directly following $\delta$ and $f \in L((q,\alpha))$ if and only if $q \in F$, we have that $\Prob_{M,\scheduler}(\LTLglobally\LTLfinally f) = \Prob_{\mathcal B}(w)$, and hence, since $\Prob_{\mathcal B}(w) > 0$ the second conjunct of $\qform$ will be satisfied.

This concludes the proof that $M \models \Phi$.

($\Leftarrow$) Suppose that $M \models \Phi$. Thus, there exists a scheduler assignment $\sassign = ((\sigma,\scheduler))$ such that 
$M,\sassign \models \big(
\exists \pi:{\sigma}.
\forall \pi':{\sigma}.\; 
\LTLglobally \bigwedge_{\alpha \in \Lambda}
(\alpha_{\pi}\leftrightarrow \alpha_{\pi'})\big) \wedge
\prob(\LTLglobally\LTLfinally f_\sigma)>0$.

Then, since $M_\sassign$ satisfies the first conjunct, there exists a path assignment $\passign$ such that $M,\sassign,\passign \models \forall \pi':{\sigma}.\; 
\LTLglobally \bigwedge_{\alpha \in \Lambda}
(\alpha_{\pi}\leftrightarrow \alpha_{\pi'})$.
Let $w \in \Lambda^\omega$ be the word obtained from the sequence of labels of the path assigned to the path variable $\pi$ in $\passign$. By the choice of $\passign$ we have that all paths in $M_{\scheduler}$ are labeled with $w$. Therefore, $\Prob_{M,\scheduler}(\LTLglobally\LTLfinally f) = \Prob_{\mathcal B}(w)$. Since $M_\sassign$ satisfies the second conjunct, we have that $\Prob_{M,\scheduler}(\LTLglobally\LTLfinally f) >0$, which thus implies that $\Prob_{\mathcal B}(w)>0$.

This concludes the proof that model checking \phl is undecidable.
\qed
\end{proof}

In the proof of Theorem~\ref{thm:undecidability-pba} we reduced the emptiness problem for PBA to the model checking of a \phl formula of the form 
$\exists\sigma\big(
(\exists \pi:{\sigma}.\forall\pi':\sigma.\; \psi) \wedge (\pexpr \bowtie c)\big)$
which contains a single existential scheduler quantifier and a \hyperctl formula  in the $\exists^1\forall^1$ fragment of \hyperltl.
An alternative undecidability proof would be to reduce the general synthesis problem for \hyperltl to the model checking problem for \phl. Intuitively, for any \hyperltl formula $\hform$ one can construct an MDP in which the actions correspond to the possible outputs of a strategy and the probabilistic choices encode the nondeterministic input from the environment. Since schedulers can be randomized, the \phl formula needs to encode the requirement that the implementation must be deterministic, which is easily done in \hyperltl. Thus, the undecidability of synthesis for different fragments of \hyperltl can be used to establish he undecidability of model checking of \phl formulas without probabilistic expressions. However, note that the synthesis problem  $\exists^1\forall^1$ fragment of \hyperltl is decidable, so the probabilistic predicate in the \phl formula in the above proof plays a key role in establishing the undecidability of the model checking for \phl formulas of this form.\looseness=-1

We saw in the previous section that several examples of interest are \phl formulas of the form 
$\forall \sigma_1\ldots\forall\sigma_n.\ 
\big((\forall \pi_1:{\sigma_1}\ldots\forall\pi_n:{\sigma_n}.
\;\psi) \rightarrow \pexpr \bowtie c\big)$. 
Analogously to Theorem~\ref{thm:undecidability-pba}, we can show that the model checking problem for \phl formulas of the form
$\exists \sigma_1\ldots\exists\sigma_n.\ 
(\forall \pi_1:{\sigma_1}\ldots\forall\pi_n:{\sigma_n}.\;\psi \wedge \pexpr \bowtie c)$ is undecidable. (The proof can be found in Appendix~\ref{sec:undecidabilityproofs}.) The undecidability for formulas of the form $\forall \sigma_1\ldots\forall\sigma_n.\ 
\big((\forall \pi_1:{\sigma_1}\ldots\forall\pi_n:{\sigma_n}.
\;\psi) \rightarrow \pexpr \bowtie c\big)$
then follows by duality. In the next two sections we present an approximate model checking procedure and a bounded model checking procedure for \phl formulas in these two classes.

Since there are finitely many deterministic schedulers with a given fixed number of states, the result stated in the next theorem is easily established.

\begin{theorem}
For any constant $b \in \mathbb{N}$, the model checking problem for \phl restricted to deterministic finite-memory schedulers with $b$ states is decidable. 	
\end{theorem}
\begin{proof}
Since there are finitely many deterministic schedulers with $b$ states, we can verify a given \phl formula by enumeration of schedulers, since for fixed schedulers both the verification of the \hyperctl part and and the verification of the formulas over probabilistic expressions are decidable.\qed
\end{proof}

\section{Approximate Model Checking}
In this section we provide a sound but incomplete procedure for model checking a fragment of \phl.
The fragment we consider consists of those \phl formulas that are positive Boolean combinations of 
formulas of the form
\begin{equation}\label{eq:universal-fragment}
\qform  =  \forall \sigma_1\ldots\forall\sigma_n.\ 
\big(\hform \rightarrow 
c_1 \cdot \prob(\pathform_1) + \ldots + c_k \cdot \prob(\pathform_k) \bowtie c
\big)
\end{equation}
where $\hform = \forall \pi_1:{\sigma_1}\ldots\forall\pi_n:{\sigma_n}.\;\psi$ and
the formula $\psi$ is such that:
\begin{itemize}
\item $\psi$ does not contain path quantifiers,
\item $\psi$ describes an $n$-safety property (intuitively, a safety property on $M^n$~\cite{Clarkson+Schneider/10/Hyperproperties}).
\end{itemize}
The formulas in Example~\ref{ex:robot}, Example~\ref{ex:bounded-performance}, and Example~\ref{ex:diff-privacy} in Section~\ref{sec:examples-logic} fall into this class (to see this, note that for Example~\ref{ex:diff-privacy} the conjunction of the probabilistic expressions can be taken outside of the scope of the scheduler quantifiers).

According to the conditions we impose above, $\psi$ contains at most one path variable associated with each scheduler variable in $\{\sigma_1,\ldots\sigma_n\}$. This will allow us to use the classical self-composition approach to obtain an automaton for $\hform$. Furthermore, requiring that $\psi$ describes an $n$-safety property will enable us to construct a deterministic safety automaton for $\hform$ which, intuitively, represents the most general scheduler in $M^n$, such that every scheduler that refines it results in a Markov chain in which all paths satisfy the formula $\psi$.

Since for every Markov chain $C$ we have $\Prob^C(\{\pi\in\Pathsi(C) \mid \pi \models \pathform\}) = 1 - \Prob^C(\{\pi\in\Pathsi(C) \mid \pi \models \neg\pathform\})$, it suffices to consider the case when $\bowtie$ is $\leq$ (or $ <$) and 
$c_i \geq 0$ for each $i = 1,\ldots,k$.
Otherwise, if $\bowtie$ is $\geq$ (or $>$) we replace each $c_i \cdot \prob(\pathform_i)$ by $c_i \cdot (1-\prob(\neg\pathform_i))$, and then move the sum $-\sum c_i \cdot \prob(\neg\pathform_i)$ to the other side of the inequality. 

When $\bowtie$ is $\leq$ (or $ <$), if for some $i$ we have $c_i < 0$ we can rewrite the probabilistic expression by replacing $c_i \cdot \prob(\pathform_i)$ by $c_i \cdot (1-\prob(\neg\pathform_i))$, which is equal to $-c_i \cdot \prob(\neg\pathform_i) + c_i$. The constant term $c_i$ is then moved to the constant expression on the right of $\bowtie$, and we are left with $d_i \cdot \prob(\pathform_i)$ where $d_i = -c_i > 0$.

Thus, it suffices to consider probabilistic predicates of the form $\sum_{i=1}^k c_i \cdot \prob(\varphi_i)\leq c$, where $c_i >0$ for each $i=1,\ldots,k$. The case when $\bowtie$ is $<$ is analogous.

\begin{figure}[t]
\centering
\scalebox{.8}[0.9]{
\begin{tikzpicture}
	\node[] (formula) at (0,-.5){$\qform  =  \forall \sigma_1\ldots\forall\sigma_n.\ 
\big((\forall \pi_1:{\sigma_1}\ldots\forall\pi_n:{\sigma_n}.\;\psi) \rightarrow 
c_1 \cdot \prob(\pathform_1) + \ldots + c_k \cdot \prob(\pathform_k) \bowtie c
\big)
$};
	\node[](MDP) at (-6.6,-.5){$M$};
	
	\node[draw, rounded corners, align = center] (safetyAuto) at (-1,-2) {construct\\ safety automaton};
	
	\path[draw,->,thick] (formula) |- (-1,-1) -| node[left]{$\psi$}(safetyAuto.north);
	
	\node[draw, rounded corners, align=center] (selfcomp) at (-4,-2){compute\\self-composition};
	\path[draw,->,thick] (MDP) |-  node[left]{}(selfcomp.west);
	\path[draw,->,thick] (-4,-.8) --  node[left]{$n$}(selfcomp.north);
	
	\node[draw, rounded corners, align = center](prodMnD) at (-2.7,-3.5){compute $M^n \otimes \mathcal D_\psi$};
	\path[draw,->,thick] (selfcomp) |- node[left]{$M^n$} (-3,-3) -|(prodMnD.120);
	\path[draw,->,thick] (safetyAuto) |- node[right]{$\mathcal D_\psi$} (-2,-3) -|(prodMnD.60);
	
	\node[draw,rounded corners, align= center] (rabinAuto) at (3,-2){construct Rabin automata};
	\path[draw,->,thick] (2,-.8) --  node[left]{$\varphi_1$}(2,-1.8);
	\path[draw,->,thick] (4.5,-.8) --  node[right]{$\varphi_k$}(4.5,-1.8);
	\node[] (dots) at (3.25,-1){...};
	
	\node[draw, rounded corners, align = center] (finalprod) at (2,-4){$\widehat{M}_\hform \otimes \mathcal A_1 \otimes \dots \otimes \mathcal A_k $};
	\path[draw,->,thick] (prodMnD) |-  node[below right]{$\widehat{M}_\hform$}(finalprod);
	\path[draw,->,thick] (2,-2.22) |- node[above left]{$\mathcal A_1$} (2,-3) -|(finalprod.160);
	\path[draw,->,thick] (4.5,-2.22) |- node[above right]{$\mathcal A_k$} (4,-3) -|(finalprod.15);
	\node[] (dots) at (3.25,-2.5){...};
	\node[] (dots) at (2.25,-3.5){...};
	
	\node[] (final) at (2,-5){$\widetilde{M}$}; 
	\path[draw,->,thick](finalprod) -- (final);
	
\end{tikzpicture}
}

\caption{Approximate model checking of \phl formulas of the form~(\ref{eq:universal-fragment}).}	
\label{fig:overapproximation}
\end{figure}
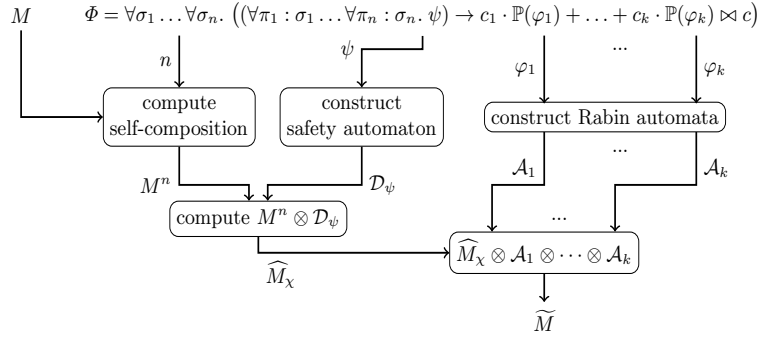

We now describe a sound procedure for checking whether an MDP $M = (S,\Act,\trans,\init,\AP,L)$ satisfies a \phl formula $\qform$ of the form~(\ref{eq:universal-fragment}). If the answer is positive, then we are guaranteed that $M\models \qform$, but otherwise the result is inconclusive. The workflow of the method is depicted in Figure~\ref{fig:overapproximation}. The method proceed in the following steps.

\begin{enumerate}

\item {\bf Deterministic safety automaton for $\psi$.} 
We begin by constructing  a deterministic safety automaton $\mathcal D_\psi = (Q,S^n,\delta,q_0)$ for the $n$-hyper safety property $\psi$. 
For the language of $\mathcal D_\psi$ we have that for each word $w\in (S^n)^\omega$ it holds that $w \in \mathcal L(\mathcal D_\psi )$ if and only if for some scheduler assignment $\sassign$ it holds that $M,\sassign,\passign_w \models \psi$, where $\passign_w$ is the path assignment corresponding to the word $w= (s_{0,1},\ldots,s_{0,n})(s_{1,1},\ldots,s_{1,n})\ldots$, formally defined as 
\[\passign_w = ((\pi_{1},s_{0,1}s_{1,1}\ldots),\ldots,(\pi_{n},s_{0,n}s_{1,n}\ldots)).\]
Note that, since $\psi$ does not contain path quantifiers, and all path variables appearing in $\psi$ are assigned in the path assignment $\passign_w$, the scheduler assignment $\sassign$ does not play any role for the satisfaction of $\psi$.

\item {\bf Product MDP for $M^n$ and $\mathcal D_\psi$.} 
As a second step we construct the $n$-self-composition MDP $M^n =  (S^n,\Act^n,\widehat\trans,\widehat\init,\AP,\widehat L)$, and then build the product of $M^n$ with the deterministic safety automaton $\mathcal D_\psi$, which is the MDP $\widehat M_\hform =  (\widehat S_\hform,\Act^n,\widehat\trans_\hform,\widehat\init_\hform,\AP,\widehat L_\hform)$ with the following components:
 \begin{itemize}
 \item $\widehat S_\hform = S^n \times Q$;
 \item $\widehat\trans_\hform((\widehat s,q),a,(\widehat s',q')) = 
 \begin{cases}
 \widehat\trans(\widehat s,a,\widehat s') & \text{if } q' = \delta(q,\widehat s'),\\
 0 & \text{otherwise};
 \end{cases}$
 \item $\widehat\init_\hform((\widehat s,q))= \begin{cases}
 \widehat\init(\widehat s) & \text{if } q=\delta(q_0,\widehat s),\\
 0 & \text{otherwise};
  \end{cases}$
 \item $\widehat{L}_\hform((\widehat{s},q)) = \bigcup_{i=1}^n\{a_{\sigma_i} \mid a \in\widehat{L}(\widehat{s})(\sigma_i)\}.$ 
 \end{itemize}
 
Note that we can assume that $\widehat M_\hform$ is a valid MDP, since if this is not the case, that is, if there are states without enabled actions, or actions for which the set of successors does not form a probability distribution, such states and actions can be iteratively eliminated until a valid MDP is obtained.
 
 Intuitively, the infinite paths in $\widehat M_\hform$ correspond to the tuples of infinite paths in $M$ such that each such tuple satisfies the  $n$-hyper safety property $\psi$.
 
\item\label{proc:overapprox} {\bf Overapproximation of $
\max_{\overline{\scheduler}= \scheduler_1 \parallel\ldots\parallel\scheduler_n}
\sum_{i=1}^k(c_i \cdot \Prob_{\widehat M_\hform,\overline{\scheduler}}(\varphi_i)).$}

Now, after constructing the MDP $\widehat M_\hform$, our goal is to check that for every scheduler assignment $\sassign$ for the MDP $M$ such that $\overline{\scheduler}= \scheduler_1 \parallel\ldots\parallel\scheduler_n\in\Sched(\widehat M_\hform)$ the inequality
$\sum_{i=1}^k(c_i \cdot \Prob_{\widehat M_\hform,\overline{\scheduler}}(\varphi_i)) \leq c$ is satisfied. That would mean, intuitively, that every scheduler assignment that satisfies $\hform$ also satisfies the above inequality, which is the property stated by $\qform$.

Note that, if we establish that $\max_{\overline{\scheduler}= \scheduler_1 \parallel\ldots\parallel\scheduler_n}
\sum_{i=1}^k(c_i \cdot \Prob_{\widehat M_\hform,\overline{\scheduler}}(\varphi_i))\leq c$, then we have established the above property. 
Computing exactly the value $\max_{\overline{\scheduler}= \scheduler_1 \parallel\ldots\parallel\scheduler_n}
\sum_{i=1}^k(c_i \cdot \Prob_{\widehat M_\hform,\overline{\scheduler}}(\varphi_i))$, however, is not algorithmically possible in light of the undecidability results established in the previous section. Therefore, we will overapproximate this value by computing a value $c^* \geq \max_{\overline{\scheduler}= \scheduler_1 \parallel\ldots\parallel\scheduler_n}
\sum_{i=1}^k(c_i \cdot \Prob_{\widehat M_\hform,\overline{\scheduler}}(\varphi_i))$ and if $c^* \leq c$, then we can conclude that the property holds.
The value $c^*$ is computed as 

\begin{equation}\label{eq:scheduler-overapprox-max}
c^* =  \max_{\widehat\scheduler \in \Sched(\widehat{M}_\hform)}
\sum_{i=1}^k(c_i \cdot \Prob_{\widehat M_\hform,\widehat\scheduler}(\varphi_i)). 
\end{equation}

Note that for the schedulers $\widehat{\scheduler}$ considered in the maximization of~(\ref{eq:scheduler-overapprox-max}) it is not in general possible to decompose $\widehat{\scheduler}$ into schedulers $\scheduler_1,\ldots,\scheduler_n \in \Sched(M)$. Therefore we have the following inequality
\[
\max_{\widehat\scheduler \in \Sched(\widehat{M}_\hform)}
\sum_{i=1}^k(c_i \cdot \Prob_{\widehat M_\hform,\widehat\scheduler}(\varphi_i)) \geq\
\max_{\overline{\scheduler}= \scheduler_1 \parallel\ldots\parallel\scheduler_n}
\sum_{i=1}^k(c_i \cdot \Prob_{\widehat M_\hform,\overline{\scheduler}}(\varphi_i)),
\]
which implies that $c^*$ has the desired property.

To compute the value defined in~(\ref{eq:scheduler-overapprox-max}) we proceed as follows.

\begin{enumerate}
\item {\bf Deterministic Rabin automata for $\pathform_1,\ldots,\pathform_k$.} According to the definition of \phl, each of the formulas $\pathform_i$ is actually an LTL formula over atomic propositions in $\AP_{\varss}$. Thus, for each $\pathform_i$ we can construct a  deterministic Rabin automaton $\mathcal A_i = (Q_i,2^{\AP_{\varss}},\delta_i,q_{i,0},(B_{i,j},G_{i,j})_{j=1}^{m_i})$, such that for every $w$ we have $w \in \mathcal{L}(\mathcal A_i)$ if and only if $w \models \pathform_i$.

\item {\bf Product MDP for $\widehat{M}_\hform$ and $\mathcal A_1,\ldots ,\mathcal A_k$.}
We now compute the product of the MDP $\widehat M_{\hform}$ constructed earlier and the automata $\mathcal A_1,\ldots ,\mathcal A_k$ which is the MDP $\widetilde M = \widehat{M}_\hform \otimes \mathcal A_1 \otimes\ldots \otimes \mathcal A_k$ defined as given below.
\[\widetilde M=\widehat{M}_\hform \otimes \mathcal A_1 \otimes\ldots \otimes \mathcal A_k = 
(\widetilde S,\Act^n,\widetilde \trans,\widetilde \init, Q_1 \times \ldots \times Q_k,\widetilde L), \text{ where} \]
\begin{itemize}
\item $\widetilde S = \widehat S_\hform \times Q_1 \times \ldots \times Q_k$;

%\bigskip

\item $\widetilde\trans ((\widehat s,q_1,\ldots,q_k),a,(\widehat s',q_1',\ldots,q_k')) = \begin{cases}
 \widehat\trans_\hform(\widehat s,a,\widehat s') & \text{if } q'_i = \delta_i(q_i,\widehat L_\hform(\widehat s'))\\& \quad\text{for }i=1,\ldots,k\\
 0 & \text{otherwise};
 \end{cases}$
 
 %\bigskip
 
 \item $\widetilde\init((\widehat s,q_1,\ldots,q_k))= \begin{cases}
 \widehat\init_\hform(\widehat s) & \text{if } q = \delta_i(q_{i,0},\widehat L_\hform(\widehat s))\text{ for }i=1,\ldots,k\\
 0 & \text{otherwise};
  \end{cases}$

%\bigskip

\item $\widetilde L((\widehat s,q_1,\ldots,q_k)) = \{(q_1,\ldots,q_k)\}$.

\end{itemize}

%\bigskip

\item {\bf Success sets for combinations of $\pathform_1,\ldots,\pathform_k$.} Now we consider each combination of formulas in $\{\pathform_1,\ldots,\pathform_k\}$, i.e., each subset $I \subseteq \{1,\ldots, k\}$ such that $I \neq \emptyset$. For each $I$, the conjunction of the accepting  conditions of the deterministic Rabin  automata $\mathcal A_i$ for $i \in I$ is 
\begin{equation}\label{eq:combined-property}
\bigwedge_{i \in I}\bigvee_{1\leq j \leq m_i} (\LTLfinally\LTLglobally \neg \widetilde B_{i,j} \wedge \LTLglobally\LTLfinally \widetilde G_{i,j}),
\end{equation}
where $\widetilde B_{i,j} = \{(s,q_1,\ldots,q_k)\in\widetilde{S}\mid q_i \in B_{i,j}\}$ and similarly for $\widetilde G_{i,j}$.

The above property can be rewritten as 
\[\bigvee_{(j_i \in \{1,\ldots,m_j\})_{i \in I}}\Big(
\big(\LTLfinally\LTLglobally \bigwedge_{i \in I} \neg \widetilde B_{i,j_i}\big) \wedge
\bigwedge_{i \in I}\LTLglobally\LTLfinally \widetilde G_{i,j_i}
\Big).\]
With this representation, we apply the methods described in~\cite{PrinciplesOfModelChecking} to compute the so called \emph{success set} $U_I\subseteq \widetilde S$ for the property~(\ref{eq:combined-property}) for each $I \subseteq \{1,\ldots, k\}$. Intuitively, in $U_I$ there exists a scheduler that can enforce the conjunction of the properties whose indices are in the set $I$.
% (and possibly some of the remaining properties but this is not of importance).

%\bigskip

\item {\bf Linear program for optimal scheduler.} Finally, in order to compute the value defined in~(\ref{eq:scheduler-overapprox-max})
we solve the following linear program.

\begin{equation}\label{eq:optimization-problem}
\text{minimize } \sum_{\widetilde s \in \widetilde S} x_{\widetilde s} \text{ subject to:}
\end{equation}

\[
\begin{array}{llll}
x_{\widetilde s}  &\geq & 0 &\text{ for all } \widetilde s \in \widetilde S\\&&\\
x_{\widetilde s}  &\geq & \sum_{i \in I} c_i &\text{ for all } I \subseteq \{1,\ldots, k\}\text{ and } \widetilde s \in U_I \\&&\\
x_{\widetilde s} &\geq& \sum_{\widetilde t \in \widetilde S} \trans(\widetilde s,a,\widetilde t) \cdot x_{\widetilde t} &\text{ for all } \widetilde s \in \widetilde S \text{ and } a \in \Act^n.
\end{array}
\]
Let $(x_{\widetilde s}^*)_{s\in \widetilde S}$ be the optimal solution of the linear program~(\ref{eq:optimization-problem}), and let 
\[c^* = \sum_{\widetilde s \in\widetilde S} \widetilde\init(\widetilde s) \cdot x_{\widetilde s}^*.\]

\end{enumerate}

\item {\bf Checking $M \stackrel{?}{\models}\qform.$}
  If $c^* \leq c$, then for all tuples of schedulers $\scheduler_1,\ldots,\scheduler_n$ we have that if 
$M_{\scheduler_1 \parallel\ldots\parallel\scheduler_n} \models\hform$, then for their product $\overline{\scheduler}= \scheduler_1 \parallel\ldots\parallel\scheduler_n$ it holds that $\sum_{i=1}^k
(c_i \cdot \Prob_{M^n,\overline{\scheduler}}(\varphi_i))\leq c$, and we conclude that $M \models \qform$.

If, on the other hand, we have  $c^* > c$, then the result is inconclusive.

\end{enumerate}

In the appendix we establish the correctness of the model checking procedure for establishing $M \models \qform$ for formulas of the form~(\ref{eq:universal-fragment}), stated in the next theorem.

\begin{theorem}[Correctness]\label{thm-correctness-overapprox}
Let $\widetilde M =  \widehat{M}_\hform \otimes \mathcal A_1 \otimes\ldots \otimes \mathcal A_k$ be the MDP constructed above, and let $(x_{\widetilde s}^*)_{s\in \widetilde S}$ be the optimal solution to the linear program~(\ref{eq:optimization-problem}). If $\bowtie \in \{\leq, <\}$, then it holds that
\[\sum_{\widetilde s \in\widetilde S} \widetilde\init(\widetilde s) \cdot x^*_{\widetilde s} \bowtie c \;\text{ implies }\ M\models\forall \sigma_1\ldots\forall\sigma_n.\ \big(\hform \rightarrow \sum_{i=1}^k c_i \cdot \prob(\pathform_i) \bowtie c\big).\]
\end{theorem}

The success set $U_I$ for each of the sets $I \subseteq \{1,\ldots,k\}$ can be computed in time polynomial in the size of $\widetilde M$. Since we can include in (\ref{eq:optimization-problem}) for each $\widetilde s \in \widetilde S$ only the inequality $x_{\widetilde s}  \geq  \sum_{i \in I} c_i$ with the largest $\sum_{i \in I} c_i$, the size of the constraint system in (\ref{eq:optimization-problem}) is polynomial in the size of $\widetilde M$. Thus, we have the following result.

\begin{theorem}[Complexity]\label{thm-complexity-overapprox}
Given an MDP $M = (S,\Act,\trans,\init,\AP,L)$ and a \phl formula $\qform$ of the form~(\ref{eq:universal-fragment}) the model checking procedure above runs in time polynomial in the size of $M$ and doubly exponential in the size of~$\qform$. 
\end{theorem}

Note that when $M\not\models\qform$ the result of the above procedure will be inconclusive. In the next section we present a bounded model checking procedure which can be used to search for counterexamples to \phl formulas of the form~(\ref{eq:universal-fragment}).
\label{sec:algorithm-overapprox}

\section{Bounded Model Checking}
We present a bounded model-checking procedure for \phl formulas of the form
\begin{equation}\label{eq:existential-fragment}
\qform  = \exists \sigma_1\ldots\exists\sigma_n.\ 
\big(\hform \wedge 
c_1 \cdot \prob(\pathform_1) + \ldots + c_k \cdot \prob(\pathform_k) \bowtie c
\big)
\end{equation}
where $\hform = \forall \pi_1:{\sigma_1}\ldots\forall\pi_n:{\sigma_n}.\;\psi$ is in the $\forall^*$ fragment of \hyperltl~\cite{fhlst18}. Examples of formulas in this fragment are the formulas in Example~\ref{ex:cause} and Example~\ref{ex:hard-soft}.
By finding a scheduler assignment that is a witness for a  \phl formula of the form~(\ref{eq:existential-fragment}) we can find counterexamples to \phl formulas of the form~(\ref{eq:universal-fragment}).

Given an MDP $M = (S,\Act,\trans,\init,\AP,L)$, a bound $b \in \mathbb{N}$, and a \phl formula $\qform = \exists \sigma_1\ldots\exists\sigma_n.\ \big(\hform \wedge c_1 \cdot \prob(\pathform_1) + \ldots + c_k \cdot \prob(\pathform_k) \bowtie c\big)$, the \emph{bounded model checking problem for $M,b$ and $\qform$} is to determine whether there exists a \emph{deterministic finite-memory scheduler} $\widetilde \scheduler = \scheduler_1||\dots||\scheduler_n$ for $M^n$ composed of deterministic finite-memory schedulers $\scheduler_i = ( Q^i, \delta^i,  q^i_0,\mathit{act}_i)$ for $M$ for $i\in \{1,\dots,n\} $, with $|\scheduler| = b$ such that 
$M^n_{\widetilde\scheduler} \models \hform \wedge \sum_{i=1}^k 
(c_i \cdot \prob(\pathform_i)) \bowtie c$. 

\begin{figure}[t]
\centering
\scalebox{.8}[.8]{
\begin{tikzpicture}
	\node[draw, align=center,thick](init)at(0,0){create \\ consistency constraint}; 
	\node[](chi) at (-.5,1){$\hform$};
	\path[draw,->,thick] (chi.south) to (-.5,0.5);
	\node[](M) at (.5,1){$M$};
	\path[draw,->,thick] (M.south) to (.5,0.5);
	
	\node[draw,align=center,thick](synthesis)at(0,-2.5){\hyperltl \\ synthesis};
	\path[draw,->](init) edge node [left] {$\varphi_M^\hform$} (synthesis);
	\node[](b) at (-1.5,-2.5){$b$};
	\path[draw,->,thick] (b) to (synthesis.west);
	\node[](unreal)at(2.5,-2.5) {unrealizable};
	\path[draw,->](synthesis.east) to (unreal);
	
	\node[draw,align=center,thick](selfcomp)at(4,0){construct \\ self-composition};
	\node[](n) at (3.5,1){$n$};
	\path[draw,->,thick] (n.south) to (3.5,0.5);
	\node[](M2) at (4.5,1){$M$};
	\path[draw,->,thick] (M2.south) to (4.5,0.5);
	
	\node[draw,align=center,thick](applysched)at(4,-1.5){apply scheduler};
	\path[draw](selfcomp.south) edge [->] node[ left]{$M^n$}(applysched.north);
	\path[draw](synthesis.east) edge [->] node[above left]{$\widetilde \scheduler$}(applysched.west);
	
	\node[draw,align=center,thick](mc)at(8,-2.5){probabilistic \\ model checking};
	\path[draw](applysched.east) edge [->] node[ above right]{$\widetilde C$}(mc.west);
	\node[](P)at(8,-1.5){$P\bowtie c$};
	\path[draw,->,thick] (P.south) to (8,-2);
	\node[](result)at(10,-2.5) {$\checkmark$};
	\path[draw,->](mc.east) to (result);
	\node[draw,dotted]at(8,0){$\varphi= \exists \sigma_1\ldots \exists \sigma_n.~\hform \wedge P \bowtie c$};
	
	\draw[->](mc.south) -- (8,-3.5) -- (0,-3.5)-- (synthesis.south);
	\node[]at (4,-3.8){$\overline\varphi_{\widetilde \scheduler}$};
\end{tikzpicture}
}
\caption{Bounded model checking of MDPs against \phl formulas for the form~(\ref{eq:existential-fragment}).}
\label{fig:BMCForPHL}
\end{figure}
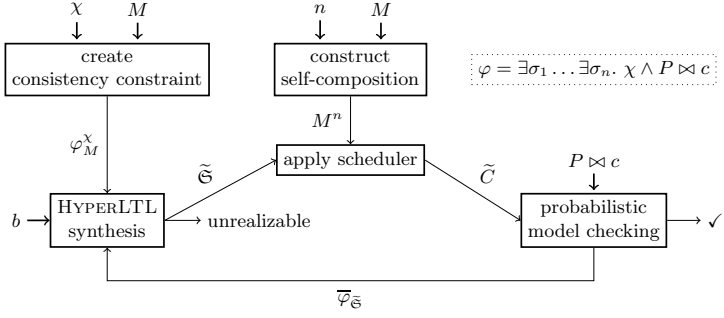

Our bounded model checking procedure employs bounded synthesis for the logic \hyperltl \cite{fhlst18} and model checking of Markov chains~\cite{probabilisticModelChecking}.  The flow of our procedure is depicted in \Cref{fig:BMCForPHL}. 
The procedure starts by checking whether there is a scheduler $\widetilde\scheduler$ for $M^n$ composed of schedulers $\scheduler_1,\dots, \scheduler_n$ for $M$ that satisfies the constraint given by the hyperproperty~$\hform$. 
This is done by synthesizing a scheduler of size $b$ for the  \hyperltl formula  $\varphi_M^{\hform}$ composed of the formula $\hform$, an encoding of $M$, which ensures  that the schedulers $\scheduler_1,\dots, \scheduler_n$ defining $\widetilde \scheduler$ follow the structure of $M$, and an additional consistency constraint that requires $\widetilde \scheduler$ to be a composition of $n$ schedulers $\scheduler_1,\dots, \scheduler_n$ for $M$. 

The formula $\varphi_M^\hform$ is constructed as follows. Let $\hform = \forall \pi_1\colon\sigma_{1}\dots \forall \pi_n\colon\sigma_{n}.~\psi$. Then $\varphi_M^\hform = \forall \pi_1 \dots \forall \pi_n.~\psi \wedge \psi_M \wedge \psi_{\mathit{consistent}}$, where 
\begin{itemize}
	\item $\psi_M$ encodes the transitions of the MDP $M$ and is given as the formula
	$$\bigwedge \limits_{1\leq i\leq n}~\bigwedge \limits_{s\in S} \LTLsquare ((s_{\sigma_i})_{\pi_i} \rightarrow \bigwedge \limits_{a \in \Act \setminus \Act_s} \neg {(a_{\sigma_i})}_{\pi_i})$$
	where $\Act_s = \{a \in \Act \mid\exists s'\in S.~\trans(s,a,s') \not =0\}$ for $s\in S$.
	
	\item $\psi_{\mathit{consistent}}$ specifies that it should be possible to decompose the synthesized scheduler to $n$ schedulers for $M$ and is given by 
	\begin{align*}
		\bigwedge \limits_{1\leq i\leq n}~\bigwedge \limits_{a\in \Act} ~& (a_{\sigma_i})_{\pi_1} =(a_{\sigma_i})_{\pi_2}  ~\LTLweakuntil~ (s_{\sigma_i})_{\pi_1} \not = (s_{\sigma_i})_{\pi_2}. \\
	\end{align*}
		
\end{itemize}

If $\varphi_M^\hform$ is realizable, then the procedure proceeds by applying the synthesized scheduler $\widetilde \scheduler$ to the 
$n$-self-composition of the MDP $M$, which results  in a Markov chain $\widetilde C= M^n_{\widetilde\scheduler}$. To check whether the synthesized scheduler also satisfies  the probabilistic constraint $P\bowtie c$, we apply a probabilistic model checker to the  Markov chain $\widetilde C$ to compute for each $\pathform_i$ the probability $\Prob_{\widetilde C}(\pathform_i)$, and then we evaluate the probabilistic predicate $P \bowtie c$. If $\widetilde C $ satisfies $P \bowtie c$, then $M^n_{\widetilde\scheduler} \models \hform \wedge \sum_{i=1}^k (c_i \cdot \prob(\pathform_i)) \bowtie c$, implying that $M\models\qform$. If not, we return back to the synthesizer to construct a new scheduler. In order to exclude the scheduler $\widetilde \scheduler$ from the subsequent search, a new constraint $\overline\varphi_{\widetilde \scheduler}$ is added to $\varphi_M^\hform$. The formula $\varphi_{\widetilde \scheduler}$ imposes the requirement that the synthesized scheduler should be different from  $\widetilde \scheduler$. 
This process is iterated until a scheduler that is a witness for $\qform$ is found, or all schedulers within the given bound $b$ have been checked. 
The complexity of the procedure is given in the next theorem and follows from complexity of probabilistic model checking \cite{probabilisticModelChecking} and that of synthesis for \hyperltl\cite{fhlst18}.
\begin{theorem}[Complexity]
	Given an MDP $M = (S,\Act,\trans,\init,\AP,L)$, a bound $b \in \mathbb{N}$, and a \emph{\phl} formula $\qform = \exists \sigma_1\ldots\exists\sigma_n. \hform ~\wedge~ c_1 \cdot \prob(\pathform_1) + \ldots + c_k \cdot \prob(\pathform_k) \bowtie c$, the bounded model checking problem for $M,b$ and $\qform$ can be solved in time polynomial in the size of $M$, exponential in $b$, and exponential in the size of~$\qform$. 
\end{theorem}
\label{sec:algorithm-bounded}
\subsection{Evaluation}\label{sec:experiments}
\begin{table}[t]
\caption{Experimental results from model checking plan non-interference.}
\begin{center}
\begin{tabular}{|l||c|c|c|c|}
	\hline 
	Benchmark & MDP Size & \# Iterations. & Synthesis time (s) &  Model checking time (s) \\
	\hline
	{Arena 3} & 16 & 6 & 12.04 & 2.68\\
	\hline
	{Arena 4} & 36 & 5 & 17.23 & 2.19\\
	\hline
	{Arena 5} & 81 & 5 & 18.49 & 2.76\\
	\hline
	{Arena 7} & 256 & 5 & 19.46 & 3.01\\
	\hline
	{Arena 9} & 625 & 7  & 168.27 & 4.72\\
	\hline
	{3-Robots Arena 4} & 36 & 9 & 556.02 & 4.5\\
	\hline
\end{tabular}
\end{center}
\label{tab:plannoninterference}
\end{table}

We developed a proof-of-concept implementation of the approach in \Cref{fig:BMCForPHL}. We used the tool BoSyHyper~\cite{fhlst18} for the scheduler synthesis and the tool PRISM~\cite{prism} to model check the Markov chains resulting from applying the synthesized scheduler to the self-composition of the input MDP. For our experiments, we used a machine with 3.3 GHz dual-core Intel Core i5 and 16 GB of memory.  

\Cref{tab:plannoninterference} shows the results of model checking the ``plan non-interference'' specification introduced in \Cref{sec:examples-logic} against MDP's representing two robots that try to reach a designated cell on grid arenas of different sizes ranging from 3-grid to a 9-grid arena. In the last instance, we increased the number of robots to three to raise the number of possible schedulers. The specification thus checks whether the probability for a robot to reach the designated area changes with the movements the other robots in the arena. In the initial state, every robot is positioned on a different end of the grid, i.e., the farthest point from the designated cell.

In all instances in \Cref{tab:plannoninterference}, the specification with $\varepsilon = 0.25$ is violated. We give the number of iterations, i.e., the number of schedulers synthesized, until a counterexample was found. 
The synthesis and model checking time represent the total time for synthesizing and model checking all schedulers. 
\Cref{tab:plannoninterference} shows the feasibility of approach, however, it also demonstrates the inherent difficulty of the synthesis problem for hyperproperties.

\begin{table}\caption{Detailed experimental results for the 3-Robots Arena~4 benchmark.}
\begin{center}
\begin{tabular}{|c||c|c|}
	\hline
	Iteration & Synthesis (s) & Model checking (s)\\
	\hline
	1 & 3.723 & 0.504\\
	\hline
	2 & 3.621 & 0.478\\
	\hline
	3 & 3.589 & 0.469\\
	\hline
	4 & 3.690 & 0.495\\
	\hline
	5 & 3.934 & 0.514\\
	\hline
	6 & 4.898 & 0.528\\
	\hline
	7 & 11.429 & 0.535\\
	\hline
	8 & 60.830 & 0.466\\
	\hline
	9 & 460.310 & 0.611\\
	\hline
\end{tabular}
\end{center}
\label{tab:tworobots}
\end{table}

\Cref{tab:tworobots} shows that  the time needed for the overall model checking approach is dominated by the time needed for synthesis: The time  for synthesizing a scheduler quickly increases in the last iterations, while the time for model checking the resulting Markov chains remains stable for each scheduler. 
Despite recent advances on the synthesis from hyperproperties~\cite{fhlst18}, synthesis tools for hyperproperties are still in their infancy. \phl model checking will directly profit from future progress on this problem.

\section{Conclusion}\label{sec:conclusion}
We presented a new logic, called \phl, for the specification of probabilistic temporal hyperproperties. The novel and distinguishing feature of our logic is the combination of quantification over both schedulers and paths, and a probabilistic operator. This makes \phl uniquely capable of specifying hyperproperties of MDPs. \phl is capable of expressing interesting properties both from the realm of security and from the planning and synthesis domains. While, unfortunately, the expressiveness of \phl comes at a price as the model checking problem for \phl is undecidable, we show how to approximate the model checking problem from two sides by providing sound but incomplete procedures for proving and for refuting universally quantified \phl formulas. We developed a proof-of-concept implementation of the refutation procedure and demonstrated its principle feasibility on an example from planning.

We believe that this work opens up a line of research on the verification of MDPs against probabilistic hyperproperties. One direction of future work is identifying fragments of the logic or classes of models that are practically amenable to verification. Furthermore, investigating the connections between \phl and simulation notions for MDPs, as well studying the different synthesis questions expressible in \phl are both interesting avenues for future work.  

\bibliographystyle{abbrv}
\bibliography{bib}

\begin{thebibliography}{10}

\bibitem{abbd20}
E.~Abraham, E.~Bartocci, B.~Bonakdarpour, and O.~Dobe.
\newblock Probabilistic hyperproperties with nondeterminism.
\newblock In {\em Automated Technology for Verification and Analysis, {ATVA}
  2020, Proc.}, 2020.

\bibitem{DBLP:conf/qest/AbrahamB18}
E.~{\'{A}}brah{\'{a}}m and B.~Bonakdarpour.
\newblock {HyperPCTL}: {A} temporal logic for probabilistic hyperproperties.
\newblock In {\em Quantitative Evaluation of Systems, {QEST} 2018, Proc.},
  2018.

\bibitem{PSL-ijcai2019}
B.~Aminof, M.~Kwiatkowska, B.~Maubert, A.~Murano, and S.~Rubin.
\newblock Probabilistic strategy logic.
\newblock In {\em Proceedings of the Twenty-Eighth International Joint
  Conference on Artificial Intelligence, {IJCAI} 2019}, pages 32--38, 2019.

\bibitem{DBLP:conf/ifm/Baier10}
C.~Baier.
\newblock On model checking techniques for randomized distributed systems.
\newblock In {\em Integrated Formal Methods, {IFM} 2010, Proc.}, volume 6396 of
  {\em LNCS}, 2010.

\bibitem{DecisionProblemsProbabilisticBuchiAutomata}
C.~Baier, N.~Bertrand, and M.~Gr{\"{o}}{\ss}er.
\newblock On decision problems for probabilistic b{\"{u}}chi automata.
\newblock In {\em {FOSSACS} 2008, Proc.}, volume 4962 of {\em LNCS}, 2008.

\bibitem{StochasticGameLogic}
C.~Baier, T.~Br{\'{a}}zdil, M.~Gr{\"{o}}{\ss}er, and A.~Kucera.
\newblock Stochastic game logic.
\newblock {\em Acta Informatica}, 49(4):203--224, 2012.

\bibitem{PrinciplesOfModelChecking}
C.~Baier and J.~Katoen.
\newblock {\em Principles of model checking}.
\newblock {MIT} Press, 2008.

\bibitem{StochasticGamesBranchingTimeObjectives}
T.~Br{\'{a}}zdil, V.~Brozek, V.~Forejt, and A.~Kucera.
\newblock Stochastic games with branching-time winning objectives.
\newblock In {\em 21th {IEEE} Symposium on Logic in Computer Science {(LICS}
  2006), Proc.}, pages 349--358. {IEEE} Computer Society, 2006.

\bibitem{Clarkson+Finkbeiner+Koleini+Micinski+Rabe+Sanchez/2014/TemporalLogicsForHyperproperties}
M.~R. Clarkson, B.~Finkbeiner, M.~Koleini, K.~K. Micinski, M.~N. Rabe, and
  C.~S{\'{a}}nchez.
\newblock Temporal logics for hyperproperties.
\newblock In {\em Principles of Security and Trust, {POST} 2014, Proc.}, volume
  8414 of {\em LNCS}, pages 265--284, 2014.

\bibitem{Clarkson+Schneider/10/Hyperproperties}
M.~R. Clarkson and F.~B. Schneider.
\newblock Hyperproperties.
\newblock {\em J. Comput. Secur.}, 18(6):1157--1210, 2010.

\bibitem{DBLP:conf/cav/CoenenFST19}
N.~Coenen, B.~Finkbeiner, C.~S{\'{a}}nchez, and L.~Tentrup.
\newblock Verifying hyperliveness.
\newblock In {\em Computer Aided Verification, {CAV} 2019, Proc.}, volume 11561
  of {\em LNCS}, 2019.

\bibitem{MDPsAndRegularEvents}
C.~{Courcoubetis} and M.~{Yannakakis}.
\newblock Markov decision processes and regular events.
\newblock {\em IEEE Transactions on Automatic Control}, 43(10), Oct 1998.

\bibitem{DifferentialPrivacy}
C.~Dwork.
\newblock Differential privacy.
\newblock In H.~C.~A. van Tilborg and S.~Jajodia, editors, {\em Encyclopedia of
  Cryptography and Security, 2nd Ed}, pages 338--340. Springer, 2011.

\bibitem{fhlst18}
B.~Finkbeiner, C.~Hahn, P.~Lukert, M.~Stenger, and L.~Tentrup.
\newblock Synthesizing reactive systems from hyperproperties.
\newblock In {\em Computer Aided Verification, {CAV} 2018, Proc., Part {I}},
  volume 10981 of {\em LNCS}, pages 289--306, 2018.

\bibitem{fhs17}
B.~Finkbeiner, C.~Hahn, and M.~Stenger.
\newblock Eahyper: Satisfiability, implication, and equivalence checking of
  hyperproperties.
\newblock In {\em Computer Aided Verification, {CAV} 2017, Proc., Part {II}},
  volume 10427 of {\em LNCS}, pages 564--570, 2017.

\bibitem{frs15}
B.~Finkbeiner, M.~N. Rabe, and C.~S{\'{a}}nchez.
\newblock Algorithms for model checking {HyperLTL} and {HyperCTL} {\^{}}*.
\newblock In {\em Computer Aided Verification - 27th International Conference,
  {CAV} 2015, Proc., Part {I}}, volume 9206 of {\em LNCS}, pages 30--48, 2015.

\bibitem{Goguen+Meseguer/1982/SecurityPoliciesAndSecurityModels}
J.~A. Goguen and J.~Meseguer.
\newblock Security policies and security models.
\newblock In {\em 1982 {IEEE} Symposium on Security and Privacy}. {IEEE}
  Computer Society, 1982.

\bibitem{InfFlowSecurity}
J.~W. Gray.
\newblock Toward a mathematical foundation for information flow security.
\newblock {\em J. Comput. Secur.}, 1(3–4):255–294, May 1992.

\bibitem{probabilisticModelChecking}
M.~Kwiatkowska, G.~Norman, and D.~Parker.
\newblock Probabilistic model checking: Advances and applications.
\newblock In {\em Formal System Verification}. Springer, 2017.

\bibitem{prism}
M.~Z. Kwiatkowska, G.~Norman, and D.~Parker.
\newblock {PRISM} 4.0: Verification of probabilistic real-time systems.
\newblock In {\em Computer Aided Verification, {CAV} 2011, Proc.}, volume 6806
  of {\em LNCS}, pages 585--591, 2011.

\bibitem{InfFlowInteractivePrograms}
K.~R. O'Neill, M.~R. Clarkson, and S.~Chong.
\newblock Information-flow security for interactive programs.
\newblock In {\em 19th {IEEE} Computer Security Foundations Workshop,
  {(CSFW-19} 2006)}, pages 190--201. {IEEE} Computer Society, 2006.

\bibitem{LTL}
A.~Pnueli.
\newblock The temporal logic of programs.
\newblock In {\em 18th Annual Symposium on Foundations of Computer Science},
  pages 46--57. {IEEE} Computer Society, 1977.

\bibitem{SabelfeldPNI}
A.~Sabelfeld and D.~Sands.
\newblock Probabilistic noninterference for multi-threaded programs.
\newblock In {\em Proceedings of the 13th {IEEE} Computer Security Foundations
  Workshop, {CSFW} '00}, pages 200--214. {IEEE} Computer Society, 2000.

\bibitem{VolpanoPNI}
D.~M. Volpano and G.~Smith.
\newblock Probabilistic noninterference in a concurrent language.
\newblock {\em J. Comput. Secur.}, 7(1), 1999.

\bibitem{StatisticalMCHyperPCTL}
Y.~Wang, M.~Zarei, B.~Bonakdarpour, and M.~Pajic.
\newblock Statistical verification of hyperproperties for cyber-physical
  systems.
\newblock {\em {ACM} Trans. Embedded Comput. Syst.}, 18(5s):92:1--92:23, 2019.

\end{thebibliography}

\newpage 
\appendix 
\section{Randomized versus Deterministic Schedulers}\label{sec:rand-sched}

Let $M =  (\{s_0,s_1,s_2\},\{\alpha_1,\alpha_2\},\trans,\init,\{a\},L)$ be the MDP shown in Figure ~\ref{fig:mdp1}.

\begin{figure}\label{fig:mdp1}
\begin{center}
\begin{tikzpicture}
\node[draw,circle] (s0) {$s_0$};
\node[draw,circle, below left= of s0] (s1) {$s_1$};
\node[draw,circle, below right= of s0] (s2) {$s_2$};
\node[left= of s1,xshift=1cm] {$L(s_1)=\{a\}$};
\node[right= of s2,xshift=-1cm] {$L(s_2)=\emptyset$};
\node[right= of s0,xshift=-1cm] {$L(s_0)=\emptyset$};
\node[left= of s0,xshift=1cm] {$\init(s_0)=1$};

\draw 	(s0) edge[left,->] node {$\alpha_1,1$} (s1)
		(s0) edge[right,->] node {$\alpha_2,1$} (s2)
		(s1) edge[loop below] node {$\alpha_1,1$} (s1)
		(s2) edge[loop below] node {$\alpha_2,1$} (s2)
;

\end{tikzpicture}
\end{center}
\caption{MDP in which labels $\alpha,p$ on arrows from $s$ to $s'$ denote $\trans(s,\alpha,s')=p$.}
\end{figure}
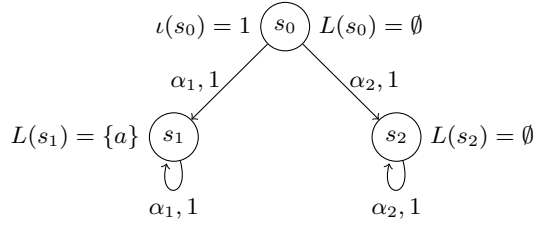

Consider the formula
$\qform=\exists\sigma.\;\exists\pi_1:\sigma.\;\exists\pi_2:\sigma.\; \LTLnext(a_{\pi_1} \wedge \neg a_{\pi_2}).$
If the quantifier is interpreted over $\Sched(M)$, then we have $M \models \qform$, as witnessed by a scheduler that in state $s_0$ chooses action $\alpha_1$ with probability $\frac{1}{2}$ and $\alpha_2$ with probability $\frac{1}{2}$. If the quantifier is interpreted over $\DSched(M)$ then $M \not\models \qform$ since for every deterministic scheduler $\scheduler$, the Markov chain $M_\scheduler$ has exactly one path.

\section{Additional Undecidability Results}
\label{sec:undecidabilityproofs}

\setcounter{theorem}{5}

\begin{theorem}\label{thm:undecidability-existential}
The model checking problem for \phl formulas of the form 
\[\exists \sigma_1\ldots\exists\sigma_n.\ 
\big((\forall \pi_1:{\sigma_1}\ldots\forall\pi_n:{\sigma_n}.\;\psi) \wedge \pexpr \bowtie c\big)\] is undecidable.
\end{theorem}
\begin{proof}
The proof follows the lines of the proof of Theorem~\ref{thm:undecidability-pba} and is again by reduction from the emptiness problem for PBA. However, the \phl formula that we will now use will assert the existence of two schedulers, one that builds up a word $w$ and another one that mimics the behaviour of the PBA on $w$. Since the schedulers are independent, this would guarantee that the 
word $w$ is chosen independently of the execution of the PBA. In order to allow for the two types of schedulers, the state space of the MDP that we will define will consists of two types of states. In states of the first type the scheduler will build up the word, and in states of the second type the scheduler will mimic the PBA.

Let $\mathcal B = (Q,\Lambda,\delta,\mu,F)$ be a PBA.

Let the MDP $M = (S,\Act,\trans,\init,\AP,L)$ have the following components.
\begin{itemize}
\item $S =  \{s_0,\widehat s_1\} \uplus \{s_\alpha \mid \alpha \in \Lambda\} \uplus (Q \times (\Lambda \uplus \{\widehat\alpha_2\}))$.
Intuitively, $s_0$ is the initial state, $\{\widehat s_1\} \uplus \{s_\alpha \mid \alpha \in \Lambda\}$ is the part of the states where the scheduler chooses a word $w$, and $Q \times (\Lambda \uplus \{\widehat\alpha_2\})$ is the part of the state space where the PBA $B$ is mimicked. 

\item $\Act = \Lambda \uplus \{\widehat\alpha_1,\widehat\alpha_2\}$, where the actions $\widehat\alpha_1$ and $\widehat\alpha_2$ serve for choosing which part of the state space will be entered.
\item $\trans : S \times \Act \times S \to [0,1]$ is defined as follows.

\begin{align*}
\trans(s_0,\alpha,s) & = \begin{cases}
1 & \text{if } \alpha=\widehat\alpha_1 \text{ and } s = \widehat s_1,\\
\mu(q) & \text{if } \alpha=\widehat\alpha_2 \text{ and } s  = (q,\widehat\alpha_2),\; q \in Q,\\
0 & \text{otherwise};
\end{cases}\\
\trans(\widehat s_1,\alpha,s') & = \begin{cases}
1 & \text{if } \alpha \in \Lambda, s' = s_\alpha,\\
0 & \text{otherwise};
\end{cases}\\
\trans(s_\alpha,\alpha',s') & = \begin{cases}
1 & \text{if } \alpha' \in \Lambda, s' = s_{\alpha'},\\
0 & \text{otherwise};
\end{cases}\\
\trans((q,\alpha),\alpha',s') & = \begin{cases}
\delta(q,\alpha',q') & \text{if } \alpha' \in \Lambda \text{ and } s' = (q',\alpha'),\; q' \in Q\\
0 & \text{otherwise};
\end{cases}
\end{align*}
Intuitively, from state $s_0$ via one of the deterministic actions $\widehat\alpha_1$ and $\widehat\alpha_2$ the corresponding part of the state space is entered. 

In states $\widehat s_1$ and $s_\alpha$ for $\alpha \in \Lambda$ the scheduler chooses the next letter and transitions via a deterministic transition to the state for that letter.

In states of the form $(q,\alpha)$ the  transitions mimic the behaviour of $\mathcal B$.
\item $\init(s_0) = 1$ and $\init(s) = 0$ for $s \neq s_0$.
\item $\AP = \{a',a'',f\} \uplus  \Lambda$. The atomic propositions $a',a''$ are used to indicate which part of the state space the state belongs to, and the proposition $f$ is used to indicate the states accepting in $\mathcal{B}$.
\item $L : S \to \AP$ is defined as follows.
\[L(s) = \begin{cases}
\emptyset & \text{if } s = s_0,\\
\{a'\} & \text{if } s = \widehat s_1,\\
\{a''\} & \text{if } s = (q,\alpha_2),\; q \in Q\setminus F,\\
\{a'',f\} & \text{if } s = (q,\alpha_2),\; q \in F,\\
\{\alpha\} & \text{if } s = s_\alpha,\\
\{\alpha\} & \text{if } s = (q,\alpha),\; q \in Q\setminus F,\; \alpha \in \Lambda, \\
\{\alpha,f\} & \text{if } s = (q,\alpha),\; q \in F,\; \alpha \in \Lambda. \\

\end{cases}
\]
\end{itemize}

We define the \phl formula
\[\begin{array}{ll}
\qform = \exists \sigma_1.\;\exists \sigma_2.\; \big(&
\forall \pi_1:{\sigma_1}.\forall \pi_2:{\sigma_2}.
(
\LTLnext a'_{\pi_{1}} \wedge 
\LTLnext a''_{\pi_{2}} \wedge 
\LTLglobally \bigwedge_{\alpha \in \Lambda}
(\alpha_{\pi_{1}}\leftrightarrow \alpha_{\pi_{2}})\big) \wedge\\ & 
\prob(\LTLglobally\LTLfinally f_{\sigma_2})>0
.
\end{array}\]

The proof that for the MDP $M$ and the formula $\Phi$ defined above it holds that there exists a word $w$ such that $\Prob_{\mathcal{B}}(w) > 0$ if and only if $M \models \Phi$ is similar to that in Theorem~\ref{thm:undecidability-pba}.

($\Rightarrow$) Let $w=\alpha_0\alpha_1 \ldots \in \Lambda^\omega$ be such that $\Prob_{\mathcal{B}}(w) > 0$.

We define the deterministic scheduler $\scheduler_1: (S \times \Act)^*S \to \Act$ such that at state $s_0$ it chooses $\widehat\alpha_1$ and for any longer sequence $s_1a_1\ldots s_n a_n s'\in (S \times \Act)^*S$ we have 
$\scheduler_1(s_1a_1\ldots s_n a_n s') = \alpha_{n-1}$. That is, the scheduler $\scheduler_1$ enters the first part of the  state space and then blindly follows the word  $w$. Thus, it is clear that all paths in $M_{\scheduler_1}$ have the same sequence of labels (corresponding to  $w$).

We define the second deterministic scheduler $\scheduler_2: (S \times \Act)^*S \to \Act$ such that at state $s_0$ it chooses $\widehat\alpha_2$ and for any longer sequence
 $s_1a_1\ldots s_n a_n s'\in (S \times \Act)^*S$ we have 
$\scheduler_2(s_1a_1\ldots s_n a_n s') = \alpha_{n-1}$. That is, the scheduler $\scheduler_2$ enters the other part of the state space and then also blindly follows the word  $w$
 regardless of the outcome of the probabilistic choices. Thus, all paths in $M_{\scheduler_2}$ also have the same sequence of labels (corresponding to the word $w$).

Therefore, by construction of $\scheduler_1$ and $\scheduler_2$, the first conjunct is satisfied.
 
Since $\trans$ is defined following $\delta$ and $f \in L((q,\alpha))$ if and only if $q \in F$, we have that $\Prob_{M,\scheduler_2}(\LTLglobally\LTLfinally f_{\sigma_2}) = \Prob_{\mathcal B}(w)$, and hence, since $\Prob_{\mathcal B}(w) > 0$ the second conjunct of $\qform$ will be satisfied.

This concludes the proof that $M \models \Phi$.

($\Leftarrow$) Suppose that $M \models \Phi$. Thus, there exists a scheduler assignment $\sassign = ((\sigma_1,\scheduler_1),(\sigma_2,\scheduler_2))$ such that 
$M,\sassign \models \big(
\big(\forall \pi_1:{\sigma_1}.\forall \pi_2:{\sigma_2}.
(
\LTLnext a'_{\pi_{1}} \wedge 
\LTLnext a''_{\pi_{2}} \wedge 
\LTLglobally \bigwedge_{\alpha \in \Lambda}
(\alpha_{\pi_{1}}\leftrightarrow \alpha_{\pi_{2}})\big) \wedge
\prob(\LTLglobally\LTLfinally f_{\sigma_2})>0$.
By the construction of the MDP $M$, and since $M_\sassign$ satisfies the first conjunct of $\qform$, there must exist at least one infinite path in the Markov chain $M_{\scheduler_1}$ visiting the state $\widehat s_1$. 
Let $w\in \Lambda^\omega$ be the word obtained from the sequence of labels on this path. Since $M_\sassign$ satisfies the first conjunct, we have that all paths in $M_{\scheduler_2}$ are labeled with $w$. Therefore, $\Prob_{M,\scheduler_2}(\LTLglobally\LTLfinally f_{\sigma_2}) = \Prob_{\mathcal B}(w)$. Since $M_\sassign$ satisfies the second conjunct, we have that $\Prob_{M,\sassign}(\LTLglobally\LTLfinally f_{\sigma_2}) >0$, which thus implies that $\Prob_{\mathcal B}(w)>0$.

This concludes the proof.

\qed
\end{proof}

\begin{corollary}\label{thm:undecidability-universal}
The model checking problem for \phl formulas of the form 
\[\forall \sigma_1\ldots\forall\sigma_n.\ 
\big((\forall \pi_1:{\sigma_1}\ldots\forall\pi_n:{\sigma_n}.
\;\psi) \rightarrow \pexpr \bowtie c\big)\] is undecidable.
\end{corollary}
\begin{proof}
The result follows from the fact that for any MDP $M$ it holds that
\[
\begin{array}{llll}
M & \models & \forall \sigma_1\ldots\forall\sigma_n.\ 
\big((\forall \pi_1:{\sigma_1}\ldots\forall\pi_n:{\sigma_n}.
\;\psi) \rightarrow \pexpr \bowtie c\big) & \text{ iff}\\
M & \not\models & \neg\forall \sigma_1\ldots\forall\sigma_n.\ 
\big((\forall \pi_1:{\sigma_1}\ldots\forall\pi_n:{\sigma_n}.
\;\psi) \rightarrow \pexpr \bowtie c\big) & \text{ iff}\\
M & \not\models & \exists \sigma_1\ldots\exists\sigma_n.\ 
\big((\forall \pi_1:{\sigma_1}\ldots\forall\pi_n:{\sigma_n}.\;\psi) \wedge \pexpr \overline\bowtie c\big),
\end{array}
\]
where 
$\overline\bowtie = 
\begin{cases}
\leq & \text{if } \bowtie  = >,\\
< & \text{if } \bowtie = \geq,\\
\geq & \text{if } \bowtie = <,\\
> & \text{if } \bowtie = \leq.\\
\end{cases}
$

\qed
\end{proof}

\section{Correctness of Approximate Model Checking}\label{sec:approx-appx}

Here we provide arguments for the correctness of the procedure described in Section~\ref{sec:algorithm-overapprox}. We begin with establishing some useful properties of the automata and MDPs constructed by the procedure.

For a path $\widetilde \pi  \in \Paths(\widetilde M)$ we denote with $\widetilde{\pi}|_{\widehat{S}_\hform}$ the projection of $\widetilde \pi$ on its first component. For each $i=1,\ldots,k$ we denote by $\widetilde{\pi}|_i$ the projection of $\widetilde \pi$ on the component corresponding to $Q_i$. By the construction of $\widetilde M$ we have that $\widetilde{\pi}|_{\widehat{S}_\hform} \in \Paths(\widehat{M}_\hform)$, and for each $i=1,\ldots,k$ the sequence $\widetilde{\pi}|_i$ is a run of $\mathcal A_i$.   
Since the automata $\mathcal A_i$ are deterministic, for every path $\widehat{\pi} \in \Paths(\widehat{M}_\hform)$ there exists a unique path $\widetilde \pi \in \Paths(\widetilde M)$ such that $\widetilde{\pi}|_{\widehat{S}_\hform} = \widehat \pi$.%, which we denote by $\widehat \pi_+$.

The next proposition establishes that the paths in the MDP $\widehat M_\hform$ are precisely the paths of the $n$-self-composition $M^n$ that satisfy the $n$-safety property $\psi$. 
\begin{proposition}\label{prop:safety-MDP}
The MDP $M^n =  (S^n,\Act^n,\widehat\trans,\widehat\init,\AP,\widehat L)$, 
the deterministic safety automaton  $\mathcal D_\psi = (Q,S^n,\delta,q_0)$ and their product 
$\widehat M_\hform =  (\widehat S_\hform,\Act^n,\widehat\trans_\hform,\widehat\init_\hform,\AP,\widehat L_\hform)$ 
defined above have the following properties. 

We have that
$\widehat \pi =  (s_{0,1},\ldots,s_{0,n},q_0)(s_{1,1},\ldots,s_{1,n},q_1)\ldots\in  \Paths(\widehat M_\hform)$ if and only if 
$\pi^n = (s_{0,1},\ldots,s_{0,n})(s_{1,1},\ldots,s_{1,n})\ldots\in \Paths(M^n)$ and 
$q_0q_1\ldots$ is a run of $\mathcal D_\psi$ on $\pi^n$. 
Furthermore, for every scheduler assignment $\sassign$ and 
path assignment  $\passign_w = ((\pi_{\sigma_1},s_{0,1}s_{1,1}\ldots),\ldots,(\pi_{\sigma_n},s_{0,n}s_{1,n}\ldots))$ such that 
$M,\sassign,\passign_w \models \psi$ and $\pi^n = (s_{0,1},\ldots,s_{0,n})(s_{1,1},\ldots,s_{1,n})\ldots \in \Paths(M^n)$ there exists a unique run $q_0q_1\ldots$ of $\mathcal D_\psi$ on $\pi^n$, and $(s_{0,1},\ldots,s_{0,n},q_0)(s_{1,1},\ldots,s_{1,n},q_1)\ldots \in \Paths(\widehat M_\hform)$.
\end{proposition}

Next we state the properties of the success sets $U_I$, which are a direct consequence from the construction of $U_I$ and standard results presented in~\cite{PrinciplesOfModelChecking}.
 
\begin{proposition}\label{prop:success-set}
In the MDP $\widetilde M =  \widehat{M}_\hform \otimes \mathcal A_1 \otimes\ldots \otimes \mathcal A_k$, each of the success sets $U_I$ for $I \subseteq \{1,\ldots,k\}$ computed above has the following properties.
\begin{itemize}
\item[(i)] There exists a finite-memory scheduler $\widetilde \scheduler \in \Sched(\widetilde M)$ s.t. for all $\widetilde s \in U_I$: \[\Prob_{\widetilde M,\widetilde\scheduler,\widetilde s}(\LTLglobally U_I \;\wedge\;\bigwedge_{\widetilde t \in U_I}\LTLglobally\LTLfinally \widetilde t) = 1.\]
\item[(ii)] For every scheduler $\widetilde \scheduler \in \Sched(\widetilde M)$ it holds that \[\Prob_{\widetilde M,\widetilde\scheduler}\{\widetilde \pi \models \LTLglobally\LTLfinally U_I \mid \forall i \in I:\ \widetilde \pi \models \bigvee_{1\leq j \leq m_i} (\LTLfinally\LTLglobally \neg \widetilde B_{i,j} \wedge \LTLglobally\LTLfinally \widetilde  G_{i,j})\} = 1.\]
\end{itemize}
\end{proposition}

Finally, we show the relationship between the optimal solution to the linear program~(\ref{eq:optimization-problem}) and the value defined in~(\ref{eq:scheduler-overapprox-max}). The construction in the proof of the below proposition is similar to the constructions presented in~\cite{MDPsAndRegularEvents}.
\begin{proposition}\label{prop:optimization}
Let $\widetilde M =  \widehat{M}_\hform \otimes \mathcal A_1 \otimes\ldots \otimes \mathcal A_k$ be the MDP constructed above. 
Then, there exists a solution $(x_{\widetilde s}^*)_{s\in \widetilde S}$ to the linear program~(\ref{eq:optimization-problem}),
and it holds that 
\[\sum_{\widetilde s \in\widetilde S} \widetilde\init(\widetilde s) \cdot x_{\widetilde s}^* = 
\max_{\widetilde\scheduler \in \Sched(\widetilde{M})}
\sum_{i=1}^k\big(c_i \cdot \Prob_{\widetilde M,\widetilde\scheduler}(\bigvee_{1\leq j \leq m_i} (\LTLfinally\LTLglobally \neg \widetilde  B_{i,j} \wedge \LTLglobally\LTLfinally \widetilde G_{i,j})\big).\]

\end{proposition}
\begin{proof}[Sketch]
To show the relationship between the value \[\max_{\widetilde\scheduler \in \Sched(\widetilde{M})}
\sum_{i=1}^k\big(c_i \cdot \Prob_{\widetilde M,\widetilde\scheduler}(\bigvee_{1\leq j \leq m_i} (\LTLfinally\LTLglobally \neg \widetilde   B_{i,j} \wedge \LTLglobally\LTLfinally \widetilde  G_{i,j})\big)\]
and the linear program~(\ref{eq:optimization-problem}) we define an MDP $\widetilde M'$ obtained from $\widetilde M$ by a construction analogous to the one presented in~\cite{MDPsAndRegularEvents}.
We construct the new MDP $\widetilde M' = (\widetilde S',\widetilde\Act',\widetilde \trans',\widetilde \init', Q_1 \times \ldots \times Q_k,\widetilde L')$ with the following components:
\begin{itemize}
\item $\widetilde S' = \widetilde S \uplus \{s^*_I \mid I \subseteq \{1,\ldots,k\}, I \neq \emptyset\}$;
\item $\widetilde \Act' = \Act^n \uplus \{a^*_I \mid I \subseteq \{1,\ldots,k\},I \neq \emptyset\}$;
\item $\widetilde \trans'(s,a,s') = \begin{cases} 
\widetilde\trans(s,a,s') & \text{if } s,s' \in \widetilde S \text{ and } a \in \Act^n,\\
1 &\text{if } s\in U_I, a = a^*_I \text{ and } s' = s^*_I \text{ for some } I \subseteq \{1,\ldots,k\},\\
1 &\text{if } s = s^*_I \text{ and } s' = s^*_I \text{ for some } I \subseteq \{1,\ldots,k\},\\
0 & \text{otherwise};
\end{cases}$
\item $\widetilde \init'(s) = \begin{cases} 
\widetilde\init(s) & \text{if } s \in \widetilde S,\\
0 & \text{otherwise};
\end{cases}$
\item $\widetilde L'(s) = \begin{cases} 
\widetilde L(s) & \text{if } s \in \widetilde S,\\
\emptyset & \text{otherwise}.
\end{cases}$
\end{itemize}
We also equip $\widetilde M'$ with a reward function $r: \widetilde S' \times \widetilde\Act' \times \widetilde S' \to \mathbb{R}_{\geq 0}$ where
\[
r(\widetilde s,a,\widetilde s') = \begin{cases} 
\sum_{i \in I} c_i & \text{if } s \in U_I, a =a^*_I \text{ and }s' = s^*_I,\\
0 & \text{otherwise}.
\end{cases}
\] 
In the MDP $\widetilde M'$ we have added an absorbing state $s^*_I$ for every $\emptyset \neq I \subseteq \{1,\ldots,k\}$ and a transition from every state in $U_I$ to $s^*_I$. Intuitively, reaching $s^*_I$ in $\widetilde M'$ corresponds to reaching and staying forever in some of the end components that are subsets of $U_I$. The reward function associates reward $r_I = \sum_{i \in I} c_i$ with each of the transitions entering $s^*_I$ from $U_I$, and all other transitions have reward $0$. Intuitively, a path in $\widetilde M'$ receives reward $r_I$ upon choosing to remain forever in the set $U_I$. We consider the maximal total expected reward in $(\widetilde M',r)$ and claim that it is equal to the quantity of interest in $\widetilde M$.
Formally, we claim that
\[
\begin{array}{l}
\max_{\widetilde\scheduler \in \Sched(\widetilde{M})}
\sum_{i=1}^k\big(c_i \cdot \Prob_{\widetilde M,\widetilde\scheduler}
(\bigvee_{1\leq j \leq m_i} (\LTLfinally\LTLglobally \neg \widetilde B_{i,j} \wedge \LTLglobally\LTLfinally \widetilde  G_{i,j})\big) 
= \\
\max_{\widetilde\scheduler \in \Sched(\widetilde{M}')}
\mathbb{E}_{\widetilde M',\widetilde\scheduler}[\sum_{j=0}^\infty r(s_j,a_j,s_{j+1})]
.
\end{array}
\]
By the definition of the linear program~(\ref{eq:optimization-problem}), the quantity on the right-hand side of the above equality is in fact the optimal solution to the linear program (\ref{eq:optimization-problem}), which follows from standard results in stochastic dynamic programming. Thus the claim of the theorem follows from the above equality, which, in turn can be established in an analogous way to the corresponding results in~\cite{MDPsAndRegularEvents} using the properties of the success sets $U_I$.
\qed
\end{proof}

Finally, combining the results above we establish the correctness of the model checking procedure for establishing $M \models \qform$ for formulas of the form~(\ref{eq:universal-fragment}).

\setcounter{theorem}{2}
\begin{theorem}[Correctness]\label{thm-correctness-overapprox}
Let $\widetilde M =  \widehat{M}_\hform \otimes \mathcal A_1 \otimes\ldots \otimes \mathcal A_k$ be the MDP constructed above, and let $(x_{\widetilde s}^*)_{s\in \widetilde S}$ be the optimal solution to the linear program~(\ref{eq:optimization-problem}). If $\bowtie \in \{\leq, <\}$, then it holds that
\[\sum_{\widetilde s \in\widehat S} \widetilde\init(\widetilde s) \cdot x^*_{\widetilde s} \bowtie c \;\text{ implies }\ M\models\forall \sigma_1\ldots\forall\sigma_n.\ \big(\hform \rightarrow \sum_{i=1}^k c_i \cdot \prob(\pathform_i) \bowtie c\big).\]
\end{theorem}

\end{document}